\RequirePackage{afterpackage}
\AfterPackage{amsthm}{
	\RequirePackage{hyperref,cleveref}
}

\documentclass[a4paper,USenglish]{lipics-v2021}

\usepackage{amssymb,amsthm,amsmath}
\usepackage{xcolor}
\usepackage{xspace}
\usepackage{hyperref}
\usepackage{float}
\usepackage{parskip}
\usepackage{caption}
\usepackage{subcaption}
\usepackage{comment}
\usepackage{bm}
\usepackage{xparse}
\usepackage{makecell}
\usepackage{thmtools, thm-restate}
\usepackage{mathtools}
\usepackage{setspace,tikz}
\usepackage[noadjust]{cite}
\usepackage{cleveref}

\usepackage{fancyhdr,graphicx,amsmath,amssymb}

\usepackage{algorithm2e}

\newcommand\commentToggle{1} 
\newcommand{\TODO}[1]{\if\commentToggle1 {\color{red}TODO #1} \fi}
\newcommand{\Jason}[1]{\if\commentToggle1 {\color{blue}Jason: #1} \fi}
\graphicspath{{figs/}}
\newcommand{\starBool}{0} 
\newcommand{\starCommand}{\if\starBool1{ {\normalfont($\star$)}}\fi}
\makeatletter
\def\thmhead@plain#1#2#3{%
	\thmname{#1}\thmnumber{\@ifnotempty{#1}{ }\@upn{#2\starCommand}}%
	\thmnote{ {\the\thm@notefont (#3)}}}
\let\thmhead\thmhead@plain
\makeatother
\newcommand{\starToggle}{\if\starBool0{\gdef\starBool{1}}\else{\gdef\starBool{0}}\fi}

\newcommand{\sameR}{\stackrel{\mathclap{\normalfont\footnotesize\mbox{R}}}{\simeq}}

\title{On inefficiently connecting temporal networks}
\author{Esteban {Christiann}}{École normale supérieure Paris Saclay, 91190 Gif-sur-Yvette, France \and \url{}}{estebanc@protonmail.ch}{}{}
\author{Eric {Sanlaville}}{Normandie Univ, UNIHAVRE, LITIS, 76600 Le Havre, France \and \url{}}{eric.sanlaville@univ-lehavre.fr}{}{DyNet RIN Tremplin Région Normandie 2020-2022}
\author{Jason {Schoeters}}{University of Cambridge, United Kingdom \and \url{https://jschoete.github.io/}}{js2807@cam.ac.uk}{}{DyNet RIN Tremplin Région Normandie 2020-2022 \& Leverhulme Trust International Professorship in Neuroeconomics}
\authorrunning{E. Christiann, E. Sanlaville, J. Schoeters} 
\Copyright{Esteban Christiann, Eric Sanlaville, Jason Schoeters} 
\acknowledgements{}
\hideLIPIcs
\newcommand\lineNumberingToggle{0} 
\if\lineNumberingToggle0 \nolinenumbers \fi

\newcommand\spannersAbstractIntro{1} 

\begin{document}
	
	\maketitle
	
	\begin{abstract}
		\if\spannersAbstractIntro1{%
		A temporal graph can be represented by a graph with an edge labelling, such that an edge is present in the network if and only if the edge is assigned the corresponding time label. A journey is a labelled path in a temporal graph such that labels on successive edges of the path are increasing, and if all vertices admit journeys to all other vertices, the temporal graph is temporally connected. A temporal spanner is a sublabelling of the temporal graph such that temporal connectivity is maintained.
		The study of temporal spanners has raised interest since the early 2000's. Essentially two types of studies have been conducted: the positive side where families of temporal graphs are shown to (deterministically or stochastically) admit sparse temporal spanners, and the negative side where constructions of temporal graphs with no sparse spanners are of importance. Often such studies considered temporal graphs with happy or simple labellings, which associate exactly one label per edge. In this paper, we focus on the negative side and consider proper labellings, where multiple labels per edge are allowed. More precisely, we aim to construct dense temporally connected graphs such that all labels are necessary for temporal connectivity.
		Our contributions are multiple: we present the first labellings maximizing a local density measure; exact or asymptotically tight results for basic graph families, which are then extended to larger graph families; an extension of an efficient temporal graph labelling generator; and overall denser labellings than previous work even when restricted to happy labellings. 
		}\else{%
		Temporal network design studies how to optimally design networks with connections over time, modelled as a temporal graph, which are often represented as a graph with an edge labelling. 
		This area of research has previously focused on minimising the number of labels to ensure temporal connectivity, whether optimising the total or local number of labels, all-to-all or terminal-to-terminal connectivity, or with additional constraints such as limiting the age of the network. 
		In this work, we focus on trying to achieve the maximum number of connections, such that all such connections are necessary for reachability. This can represent the power of an adversary trying to waste temporal network resources, even when constrained. This also implies worst-case results for temporal spanner problems, or for corresponding greedy algorithms for finding such spanners.
		Our contributions are multiple: we present denser labellings than previous work; the first dense labellings maximizing a novel measure; structural results for basic graph families, which are then extended to larger graph families; and an extension of an efficient temporal graph labelling generator. Over the course of the paper, we make some bounds meet exactly or asymptotically, and leave some interesting open questions in the conclusion.
		}\fi
	\end{abstract}
	
	\section{Introduction}
	\label{sec:intro}
	
	\if\spannersAbstractIntro1{%
	A temporal graph is a graph which can evolve over time, through the appearing and/or disappearing of edges. Numerous classical graph problems and parameters have been extended to temporal graphs, such as colouring, connected components, maximum matchings, and independent sets.\cite{marino2022coloring,vernet2022study,mertzios2020computing,hermelin2022temporal} In temporal graphs, connectivity may become very poor when considering the graph at every distinct time step, but the graph may still be connected when considering connectivity over time. Indeed, temporal connectivity is motivated through many contexts in which temporal graphs naturally arise, most notably the context of swarms of mobile entities with distance-based communication capabilities (drone networks, insect colonies, and, particularly useful during the COVID-19 pandemic: people). \cite{khanda2021efficient,dibrita2022temporal,charbonneau2013social,enright2019deleting} This temporal connectivity has since redefined classical connectivity problems, such as (temporal) dominating sets, and (temporally) connected components, and particularly interesting concerning this paper: (temporal) spanners.\cite{casteigts2021temporal,bilo2022blackout,kutner2023temporal,balev4590651temporally}
	
	After presenting a wide range of interesting changes and results concerning typical graph problems with temporal paths instead of paths, Kempe, Kleinberg, and Kumar discuss further interesting questions.\cite{KKK02} 
	One of these is whether a temporally connected graph can always be sparsified (that is, if labels can be removed) so as to obtain a ``sparse'' remaining structure maintaining temporal connectivity. Such a structure is later called a temporal spanner. Note that the static graph analogue would be asking whether a connected graph always admits a spanning tree, which is of course always the case. They follow up with a preliminary negative result, stating that some temporal graphs do not admit a linear size spanner (hypercube graphs with each edge labelled with the corresponding dimension). The real question then became whether dense temporal graphs could always admit a sparse spanner, the intuition being that there exists many more ways to potentially sparsify a dense graph.
	The question remained open for many years, until Axiotis and Fotakis answered in the negative: they construct a non-trivial dense temporal graph in which some labels may be removed but prove that a dense part has to remain to ensure temporal connectivity. \cite{axiotis2016size} A couple of years afterwards, a complementing positive result was presented by Casteigts, Peters, and Schoeters: any temporal complete graph always admits a sparse spanner.\cite{casteigts2021temporal}
	Following these, more papers surfaced related to temporal spanners: sharp thresholds on the density of random temporal graphs to asymptotically almost surely admit particular sparse spanners; positive and negative results regarding spanners which have a limited stretch, as well as on temporal spanners which are blackout-resistant.\cite{casteigts2022sharp,bilo2022sparse,bilo2022blackout}
	
	Another topic of interest in temporal graph theory is that of temporal network design, where instead of analysing a given temporal graph, one would like to design a temporal graph with some desired property or decide such a temporal graph does not exist. In most works on temporal network design, the graph itself is given and a corresponding labelling needs to constructed. 
	One of the earliest such design problems was to create a gossip protocol, that is, a schedule of pairwise communications between $n$ agents, each having some piece of information which can be transferred over successive communications, such that at the end of the schedule, all agents are up to date with all the information. It is natural to minimise the number of communications (\textit{e.g.} the total cost of phone calls), and thus some tight results arise with protocols using $2n-3$ communications, with the idea being to converge cast all information to some agent and then broadcast the information out again, designing essentially a temporal in-tree and a temporal out-tree \textit{resp.} 
	For more information, see survey \cite{zhang_broadcasting_2013}.
	More recently in \cite{mertzios2019temporal}, Mertzios \textit{et al.} reconsider and extend this work as a temporal graph design problem. Direct results from gossiping apply, but more importantly, they include other restrictions on the labelling, such as a maximum lifetime \textit{i.e.} the labels cannot be greater than some value, which was further investigated in \cite{klobas2022complexity}. Also, two measures of density, both of interest for this paper, are defined regarding a temporal graph: the temporal cost, being the total amount of labels; and the temporality, being the maximum amount of labels on some edge. The former can be seen as a global density measure, and the latter as a local one.
	Other temporal graph design problems include \cite{enright2021assigning} in which the authors aim to minimize reachability of a given graph by choosing the order of appearance of a set of given time edges, and \cite{klobas2023realizing} where the authors aim to construct a temporal graph which respects the given fastest travel times between vertices. 
	
	In this paper, we combine the study of temporal spanners and of temporal graph design, by designing dense temporal graphs such that each label is necessary for temporal connectivity.
	As opposed to most previous work we will not restrict ourselves to happy or simple labellings (one label per edge) but instead extend to consider proper labellings (multiple labels per edge allowed). This is a double-edged sword: on the one hand this intuitively may allow for much denser labellings, but on the other hand, a combinatorial explosion on the amount of possible labellings occurs implying algorithmics may be more difficult in this setting.
	In a sense, we are interested in designing the most inefficient temporal networks possible.
	Outside of the already established applications of temporal spanners in related work, the negative results in particular can have direct implications concerning adversarial behaviour in temporal network game theory and the potential waste of temporal and structural resources.\cite{liu2023adversarial,sun2022adversarial}  
	Lastly, a slowly temporally connected network may allow for time to detect any anomalies/viruses before the whole network is infected, while not hindering the supposedly essential connectivity of the network, and may have applications for fraud detection in financial transactions.\cite{rajeshwari2016real}
	
	In short, we build upon previous density measures and we aim to answer these two questions:
	\begin{itemize}
		\item 
		What is the densest temporal graph overall?
		\item 
		Given a graph class, what is the densest labelling all graphs can attain?
	\end{itemize} 
	Throughout the paper, we will steadily answer both questions simultaneously.
	}\else{%
	A fundamental part of graph theory is the design of graphs such that some useful property is respected, such as a small chromatic number or diameter, or being at least $k$-connected.\cite{heule2018computing,imase1981design,imase1983design,cai1989minimum,ishii2000minimum} Other examples include the design of good example graphs for the training of machine learning methods.\cite{lecun1989generalization} The construction of counterexamples to conjectures can be seen as graph design as well, where the property then represents the worst-case scenario regarding the conjecture. Examples of this include Grundy colouring, representing the worst case of greedy colouring, busy beavers, representing the worst amount of space that can be covered by a finite state machine, or more recently, a planar graph needing four pages for any book embedding.\cite{zaker2006results,michel2009busy,bekos2020four}
	
	
	Of course, the natural minimisation of some parameters occurs in temporal graph theory as well. A temporal graph is a graph which can evolve over time, through the appearing and/or disappearing of edges. Numerous classical graph problems and parameters have been extended to temporal graphs, such as colouring, connected components, maximum matchings, and independent sets.\cite{marino2022coloring,vernet2022study,mertzios2020computing,hermelin2022temporal} In temporal graphs, connectivity may become very poor when considering the graph at every distinct time step, but the graph may still be connected when considering connectivity over time. Indeed, temporal connectivity is motivated through many contexts in which temporal graphs naturally arise, most notably the context of swarms of mobile entities with distance-based communication capabilities (drone networks, insect colonies, and, particularly useful during the COVID-19 pandemic, human communities). \cite{khanda2021efficient,dibrita2022temporal,charbonneau2013social,enright2019deleting} This temporal connectivity has since redefined classical connectivity problems, such as (temporal) spanners, (temporal) dominating sets, and (temporally) connected components.\cite{casteigts2021temporal,bilo2022blackout,kutner2023temporal,balev4590651temporally}
	
	In \cite{zhang_broadcasting_2013}, the authors aim to create a gossip protocol, that is, a schedule of pairwise communications between $n$ agents, each having some piece of information which can be transferred over successive communications, such that at the end of the schedule, all agents are up to date with all the information. It is natural to minimise the number of communications (\textit{e.g.} the total cost of phone calls), and thus some tight results arise with protocols using $2n-3$ communications, with the idea being to converge cast all information to some agent and then broadcast the information out again, designing essentially a temporal in-tree and a temporal out-tree \textit{resp.} This can be improved slightly to $2n-4$ if instead of converge/broadcasting to/from one agent, a cycle of four agents is considered.
	More recently in \cite{mertzios2019temporal}, Mertzios \textit{et al.} reconsider the problem under the guise of temporal graphs, effectively creating a temporal graph design problem. A graph is given, and a labelling has to be assigned to the edges so as to obtain a temporally connected temporal graph, while minimising the amount of labels. Direct results from gossiping apply, but more importantly, they include other restrictions on the labelling, such as a maximum lifetime \textit{i.e.} the labels cannot be greater than some value, which was further investigated in \cite{klobas2022complexity}. Also, two measures of density, both of interest for this paper, are defined regarding a temporal graph: the temporal cost, being the total amount of labels; and the temporality, being the maximum amount of labels on some edge. The former can be seen as a global density measure, and the latter as a local one.
	Other temporal graph design problems include \cite{enright2021assigning} in which the authors aim to minimize reachability of a given graph by choosing the order of appearance of a set of given time edges, and \cite{klobas2023realizing} where the authors aim to construct a temporal graph which respects the given fastest travel times between vertices. 
	
	In these temporal graph design problems, the number of labels is minimised (or fixed and given). In this paper, we are interested in exploring the other side of the problem, when the number of labels is maximised instead. Consider the problem of, given a graph, assigning a labelling to the edges such that the resulting temporal graph is temporally connected, while maximising the amount of labels. For natural reasons, we consider maximising the amount of necessary labels, which as the term suggests, are necessary for some vertex reachability in the temporal graph. In some sense, we are interested in designing the most inefficient temporal networks possible, similarly to how Grundy colouring and busy beaver problems aim to maximise cost values which are naturally minimised in most problems. This has direct implications concerning adversarial behaviour in temporal network design, and the potential waste of temporal and structural resources.\cite{liu2023adversarial,sun2022adversarial}  Also, results on maximum density of temporally connected graphs imply worst-case scenarios for temporal spanners, as well as qualitative results on greedy algorithms for such spanners.\cite{casteigts2021temporal}
	Lastly, a slowly temporally connected network may allow for time to detect any anomalies/viruses before the whole network is infected, while not hindering the supposedly essential connectivity of the network, and may have applications for fraud detection in financial transactions.\cite{rajeshwari2016real}
	
	We build upon previous density measures and we aim to answer these two questions:
	\begin{itemize}
		\item 
		What is the densest labelling overall?
		\item 
		Given a graph class, what is the densest labelling any graph can attain?
	\end{itemize} 
	Throughout the paper, we will steadily answer both questions simultaneously.
	}\fi
	
	\subsection{Contributions}
	
	First, in \Cref{sec:preliminaries}, we give standard graph theory and temporal graph theory notation and define our setting as well as a global and a local density measure. Lower bounds from the literature and upper bounds through analysis are presented.
	Then, in \Cref{sec:trees}, we focus on tree graphs for which we obtain tight results on dense labellings through an argument on bridge edges. These results do not beat aforementioned lower bounds however. 
	Thus, in \Cref{sec:better_lower_bounds}, we present first an ad-hoc graph with a specific labelling so as to beat both lower bounds (significantly for the local measure and insignificantly for the global one). Another labelling is presented, this time on cycle graphs, which significantly beats the lower bound for the global measure as well.
	In \Cref{sec:cycles}, we decide to focus on cycles, partly due to the latter labelling showing promise for obtaining even denser labellings. For this, we decide to extend labelling generator STGen so as to fit to our setting and to cycles specifically. After executing it on small cycles, we obtain the intuition for a complex labelling which beats the lower bounds from \Cref{sec:better_lower_bounds} for local density by a factor of $1.5$, and for global density by exactly 1 label.
	Ultimately, in \Cref{sec:cacti}, we combine our previous results to obtain a labelling for cactus graphs. The density of the labelling depends highly on the largest cycle of the cactus graph. 
	We summarise and extend results in \Cref{sec:conclusion}, and conclude with some open questions and possible future work, along with some preliminary results. This includes considering different types of labellings than the one considered in this paper, as well the computational complexity of related problems.
	
	\section{Preliminaries}
	\label{sec:preliminaries}
	
	A graph $G= (V, E)$ is defined according to vertex set $V$ and edge set $E \subseteq \binom{V}{2}$. In this paper, all graphs are simple and undirected (except for the reachability graph defined below). 
	A temporal graph is a graph which can change over time, often modelled as $(G, \lambda)$ with graph $G$, called the footprint or underlying graph, and edge labelling $\lambda : E \rightarrow 2^\mathbb{N}$. The labels correspond to when the edges are present over the lifetime of the temporal graph. A pair $(e, \ell)$ with $e \in E$ and $\ell \in \lambda(e)$ is called a temporal edge, or contact. Reachability in temporal graphs is defined through temporal paths, also called journeys, which are incident temporal edges $j = (c_1, c_2, ..., c_k)$ such that for all $c_i = (e_i, \ell_i)$ with $i \in [2, k]$ we have that $\ell_i > \ell_{i-1}$. In other words, journeys obey the chronological order of time. A journey's length is the number of contacts it is composed of. If $e_1 = \{u, v\}$ and $u \not\in e_2$, and similarly $e_k = \{w, x\}$ and $x \not\in e_{k-1}$, then we say $u$ can reach $x$, or $x$ can be reached by $u$, through $j$. A journey $\mathcal{J}$ is said to cover a set of vertices $V$ if for all vertices $v$ in $V'$, $v$ is part of some contact of $\mathcal{J}$.
	For a label $\ell$ of temporal graph $\mathcal{G}$, $\mathcal{G}^{-\ell}$ corresponds to $\mathcal{G}$ without label $\ell$ (if other labels exist in $\mathcal{G}$ with the same value, then these remain; in other words, $\ell$ only represents a single label on a single edge in $\mathcal{G}^{-\ell}$). 
	
	A tree $T = (V, E)$ is a connected graph with $|E| = n-1$. A temporal branching $\mathcal{B} = (T, \lambda)$ with root $r$ is a tree $T$ with $|\lambda| = n-1$ such that vertex $r$ can reach all vertices, or equivalently for all incident edges $e_1$ and $e_2$, if edge $e_2$ is distanced further from root $r$ than edge $e_1$, then $\lambda(e_1) < \lambda(e_2)$. A tree $T = (V', E')$ of a graph $G = (V, E)$ is a subgraph of $G$ which is a tree, and it is spanning if $V' = V$. Similarly, a temporal branching $\mathcal{B} = (T, \lambda')$ with root $v$ of a temporal graph $\mathcal{G} = (G, \lambda)$ is a temporal subgraph of $\mathcal{G}$ which is a temporal branching, and it is spanning if $V(T) = V(G)$.
	It is possible to compute a temporal branching, or the reachability, of vertex $v$ by applying a search algorithm which, unlike Breadth-First Search or Depth-First Search, will prioritize the earliest edges incident to the discovered vertex set. For a more efficient version of this algorithm similar to Dijkstra's algorithm, see \cite{xuan2003computing}.
	
	
	The reachability graph $R(\mathcal{G})$ is defined on the same vertex set as $\mathcal{G}$ and an arc exists from $u$ to $v$ if and only if $u$ can reach $v$ in $\mathcal{G}$.
	If all vertices can reach all other vertices in $\mathcal{G}$, we say $\mathcal{G}$ is temporally connected. Note that a temporal graph is temporally connected if and only if the corresponding reachability graph is complete (with arcs in both directions).
	From \cite{casteigts2022simple}, two temporal graphs $\mathcal{G}_1$ and $\mathcal{G}_2$ are reachability-equivalent\footnote{Originally \textit{closure-equivalent}, but changed to reachability-equivalent for journal version (private message).} if reachability graphs $R(\mathcal{G}_1)$ and $R(\mathcal{G}_2)$ are isomorphic, denoted $\mathcal{G}_1 \sameR \mathcal{G}_2$. 
	
	A labelling is proper when no incident edges share a same label. This setting is used in many contexts, \textit{e.g.} historically for phone calls when at most one number could be phoned at the same time. It conveniently avoids the question of whether journeys are allowed to use multiple incident contacts having the same label. In \Cref{sec:conclusion} we briefly go over other types of labellings (see \cite{casteigts2022simple} for a study on different types of labellings), but for the rest of the paper, the labellings are supposed proper. A labelling is globally proper if no edges share a same label, and is incremental if all labels between 1 and $|\lambda|$ are used. Note that incremental implies globally proper, which implies proper. 
	Using terms from \cite{akrida}, a label $\ell$ in a temporal graph $\mathcal{G}$ is redundant if and only if it can be removed from $\mathcal{G}$ without modifying reachability, \textit{i.e.} $\mathcal{G} \sameR \mathcal{G}^{-\ell}$. Conversely, a label $\ell$ of $\mathcal{G}$ is necessary if and only if $\mathcal{G} \not\sameR \mathcal{G}^{-\ell}$. Note that a label is either redundant or necessary. A minimal labelling contains only necessary labels. Note that if a labelling is not minimal, it can be modified (by removing redundant labels) in polynomial time to become minimal\footnote{Many distinct minimal labellings may result from this technique, depending on the removal order.}. A temporal graph with a proper (\textit{resp.} minimal) labelling is a proper (\textit{resp.} minimal) temporal graph.
	
	The aim of this paper is to understand and construct dense temporally connected temporal graphs. In \cite{mertzios2019temporal}, the authors defined two measures of density for a temporal graph $\mathcal{G}$: the temporal cost $T(\mathcal{G})$, which is the total amount of labels in $\mathcal{G}$; and the temporality $\tau(\mathcal{G})$ which is the maximum amount of labels on an edge, among all edges. The former is intended as a global density measure, whereas the latter is more of a local one, potentially of interest for example in distributed or parallel computing. 
	The authors focussed on several connectivity constraints, among others temporal connectivity, as well as additional constraints such as restraining the lifetime of the temporal graph.
	In this paper, we only focus on temporal connectivity. 
	We thus adapt the measure of temporal cost in the following suiting manner. 
	\begin{itemize}
		\item Let $T^+(G)$ be the maximum temporal cost of graph $G$, \textit{i.e.} the maximum temporal cost $T(\mathcal{G} = (G, \lambda))$ of all proper minimal labellings $\lambda$ such that $\mathcal{G}$ is temporally connected;
		\item Let $T^+(\texttt{Class})$ be the maximum temporal cost of graph class \texttt{Class}, \textit{i.e.} the maximum value $x$ such that for all graphs $G$ of \texttt{Class}, $T^+(G) \geq x$;
		\item Let $T^+$ be the maximum temporal cost, \textit{i.e.} the maximum temporal cost $T^+(G)$ among all graphs $G$ on $n$ vertices.
	\end{itemize}
	
	Informally, the maximum temporal cost of a graph is the densest the given graph can be, the maximum temporal cost of a graph class the densest all graphs of the class can be, and the maximum temporal cost the densest any graph can be. 
	The three types of maximum temporality are defined analogously.
	Note that in order to avoid the trivial and unsatisfactory solution of assigning all natural numbers to all edges so as to obtain the densest labelling possible (with all labels being redundant), we naturally restrained the setting to minimal labellings.
	
	Since a labelling can assign no labels to any edge, a density result for class \texttt{C} translates as a lower bound for any class $\texttt{C}'$ if for all graphs $G' \in \texttt{C}'$, there exists $G \in C$ such that $G$ is an edge-deleted subgraph of $G'$. We say that such a class $\texttt{C}'$ subgraph dominates class \texttt{C}. For example, density results for \texttt{Trees} imply lower bounds for \texttt{Complete}, the class of complete graphs, since a labelling could simply ignore most edges and consider only a spanning tree of the complete graph.
	Reversely, a density result for class \texttt{C} implies an upper bound for any superclasses of \texttt{C}. For example, density results for \texttt{Cycles} are an upper bound for \texttt{Cacti}, because if by contradiction any cactus graph can be denser than cycles, then the same labelling can be used on cycles since they are cactus graphs. 
	Together, this means that density results directly transfer from one class \texttt{C} to superclass $\texttt{C}'$ if $\texttt{C}'$ subgraph dominates \texttt{C}. For example, density results for \texttt{Trees} directly transfer to \texttt{Connected}, the class of connected graphs. 
	These observations also mean that $T^+ = T^+(\texttt{Complete})$ and $\tau^+ = \tau^+(\texttt{Complete})$ unless the order of the graph $n$ is important (parity, primality, \textit{etc.}) in which case $T^+ = T^+(\texttt{S})$ and $\tau^+ = \tau^+(\texttt{S})$ for some subclass $\texttt{S}$ of \texttt{Complete}.
	
	We study three simple graph classes in this work, being \texttt{Trees} (\Cref{sec:trees}), \texttt{Cycles} (\Cref{sec:cycles}), and \texttt{Cacti} (\Cref{sec:cacti}), and in \Cref{sec:conclusion} superclasses are discussed.
	
	\subsection{Lower and upper bounds on $T^+$ and $\tau^+$}
	
	Kempe, Kleinberg and Kumar in \cite{KKK02} provide the following lower bound on the maximum temporal cost $T^+$. Consider a hypercube graph of order $n$ where the vertices represent all $\log n$-bit strings, and in which vertices $u$ and $v$ admit an edge of label $k \leq \log n$ when they agree on all bit positions but position $k$. The resulting temporal graph is minimal, and thus $T^+ \geq \tfrac{1}{2} n \log n$. With this preliminary result, they ask the open question whether denser graphs could admit denser minimal labellings. The question remained unanswered for 16 years until Axiotis and Fotakis in \cite{axiotis2016size} present a non-trivial construction of a dense and minimal temporal graph, the main idea being to make all labels of a clique of size $\tfrac{n}{3}$ necessary. The size of their construction is at most $\tfrac{1}{18} n^2 + \tfrac{3}{2}n + O(1)$, and is thus a lower bound for the maximum temporal cost $T^+$. (The linear term doesn't usually matter much, but we improve upon it later in the paper.)
	
	Contrary to $T^+$ and to the best of the authors' knowledge, no previous work has focussed on lower bounds for the maximum temporality $\tau^+$. There's the trivial lower bound of two, 
	which can easily be attained on the path graph of order three $P_3$ on vertices $v_1$, $v_2$ and $v_3$: each edge needs to have at least one label for temporal connectivity, furthermore if using only one label on each edge, then one of the labels will be larger than the other, meaning either $v_1$ cannot reach $v_3$ or vice versa, so another label is needed. 
	This lower bound is often attained in related work, for example: gossiping protocols such as in \cite{baker1972gossips, hajnal1972cure} use an edge twice (if not then the graph is of temporality 1 which is NP-hard to recognize\cite{gobel_label-connected_1991}); in \cite{mertzios2019temporal} the example minimal labelling for the cycle has temporality two (another labelling for the cycle is presented with temporality $n$, although this depends highly on the additional age restriction they consider and the graph being directed).
	
	The upper bounds are both obtained through the following observation:
	%
	%
	
	\begin{observation}
		\label{observation:label_not_in_branching_means_redundant}
		
		Take any $n$ spanning temporal branchings of temporal graph $\mathcal{G}$ such that all roots are distinct. Any label in $\mathcal{G}$ which is not part of any of these temporal branchings is redundant, as removing it doesn't change the branchings and thus doesn't affect reachability.
	\end{observation}
	
	
	Note that \Cref{observation:label_not_in_branching_means_redundant} does not prove that all labels which are part of the temporal branchings are necessary. In fact, it is possible to have temporally connected graphs (and thus $n$ spanning temporal branchings with distinct roots) with only redundant labels, for example the temporal graph consisting of the complete graph with label 1 on each edge. What \Cref{observation:label_not_in_branching_means_redundant} does imply is that all necessary labels must be part of (some of) these temporal branchings. 
	This has some nice implications for minimal temporally connected graphs. 
	
	
	\begin{lemma}
		\label{lemma:min_tc_graph_equals_its_span_temp_branch}
		A minimal temporally connected graph $\mathcal{G}$ equals the union of any $n$ spanning temporal branchings with distinct roots of $\mathcal{G}$.
	\end{lemma}
	
	\begin{proof}
		By contradiction, suppose that the union of some $n$ spanning temporal branchings with distinct roots of $\mathcal{G}$ does not equal $\mathcal{G}$. This implies at least some label of $\mathcal{G}$ isn't part of the branchings, which by  \Cref{observation:label_not_in_branching_means_redundant} means it is redundant. However, $\mathcal{G}$ is minimal so no redundant labels exist which is a contradiction.
	\end{proof}
	
	Thus, in order to create a dense but minimal temporally connected graph, we need to create spanning temporal branchings which interfere as little as possible with each other, \textit{i.e.} make sure they are as disjoint as possible.
	
	\begin{theorem}
		\label{theorem:Tplus_upperbound}
		The maximum temporal cost $T^+ \leq n^2 - n - 1$.
	\end{theorem}
	
	\begin{proof}
		By \Cref{lemma:min_tc_graph_equals_its_span_temp_branch}, a minimal temporally connected graph equals the union of any of its $n$ distinct-root spanning temporal branchings. Thus, the worst-case scenario for the total number of labels in such a graph is when these temporal branchings are all disjoint. This results in a labelling using $n-1$ labels for each branching (as they are spanning), of which there are $n$, resulting in a total of $n^2-n$ labels.
		
		Consider however the smallest label $\ell^-$ used in the graph, say on edge $e = \{u, v\}$. This label can only be part of the spanning temporal branching of root $u$, denoted $\mathcal{B}_u$, or of the spanning temporal branching of root $v$, denoted $\mathcal{B}_v$, since it's unreachable from any other vertex. Suppose \textit{w.l.o.g.} $\ell^-$ is part of $\mathcal{B}_u$. We know $v$ must reach $u$ through some journey in $\mathcal{B}_v$ arriving at some time $\ell$. Note that $\ell$ can be removed from $\mathcal{B}_v$, and $\ell^-$ added. Indeed, for all $w$, any journey $v \leadsto w$ in $\mathcal{B}_v$ is either maintained by the swap, or passes through $u$ earlier with $\ell^-$, meaning $\mathcal{B}_v$ remains a spanning temporal branching. Thus, label $\ell^-$ can be considered part of two spanning temporal branchings, decrementing the total amount of labels to $n^2 - n - 1$.  
		
		
	\end{proof}
	
	\begin{theorem}
		\label{theorem:tauplus_upperbound}
		The maximum temporality $\tau^+ \leq n-1$.
	\end{theorem}
	
	\begin{proof}
		By \Cref{lemma:min_tc_graph_equals_its_span_temp_branch}, a minimal temporally connected graph equals the union of any of its $n$ distinct-root spanning temporal branchings. Thus, the worst-case number of labels on an edge in such a graph is when the spanning temporal branchings are all label-disjoint  and all use one same edge $e=\{u, v\}$, resulting in an edge having $1$ label for each branching, of which there are $n$, resulting in a total of $n$ labels.
		
		Note however that the label from the branching corresponding to root $u$, and the label from branching corresponding to root $v$, are necessarily the same label, since otherwise the later of the two would be redundant, as both branchings can use the earlier label. Thus edge $e$ would have $n-1$ labels.
	\end{proof}

	In the rest of the paper, we improve upon the lower bounds. These eventually meet the provided upper bounds in an asymptotic manner, meaning they are tight up to a constant factor. We obtained these results by analysing the maximum temporal costs and temporality of simple families of graphs, the first of which being tree graphs.
	
	\section{Tree graphs}
	\label{sec:trees}
	
	In this section we prove the following labelling is densest possible for tree graphs (be it considering maximum temporal cost $T^+$ or maximum temporality $\tau^+$). This labelling was inspired by the pivot technique used in \cite{casteigts2021temporal} which in turn bears resemblance to Kosaraju-Sharir's algorithm for directed graphs \cite{sharir1981strong}. It also resembles gossiping strategies from \cite{baker1972gossips, hajnal1972cure}.
	
	\begin{definition}[Pivot labelling of tree graph $G$]
		See also \Cref{fig:pivot_labelling}.
		Let some arbitrary vertex $p$ be the pivot, and let integer $\ell = 1$. Treat the other vertices of $G$ in a reverse breadth-first search order from $p$ (\textit{i.e.} treating the furthest away vertices from $p$ first) in the following manner. For currently treated vertex $v$, determine the path to $p$ in $G$. On the edge of this path closest to $v$, assign label $\ell$ and then increment $\ell$.
		Let vertex $v$ be treated. Upon having treated all vertices, let $v'$ be the last vertex treated. Note that $\ell$ now equals $n$.		
		Now repeat the process but in a normal breadth-first search order from $p$, continuing with $\ell = n$.
		Finally, remove the largest label on edge $\{v', p\}$.
	\end{definition}
	
	\begin{figure}[h]
		\begin{subfigure}{.5\textwidth}
			\includegraphics[width=\textwidth]{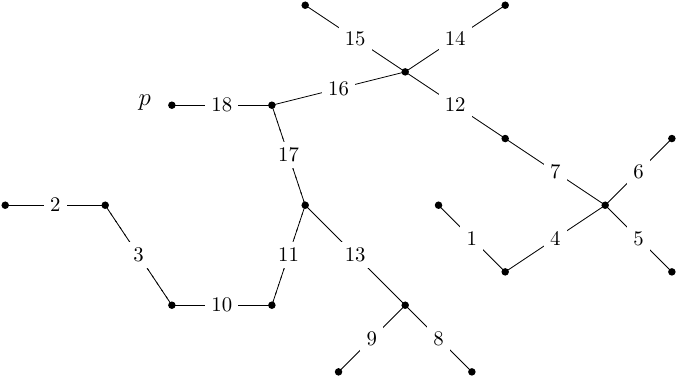}
			\caption{First (reverse) breadth-first search labelling.}
		\end{subfigure}
		\hfill
		\begin{subfigure}{.5\textwidth}
			\includegraphics[width=\textwidth]{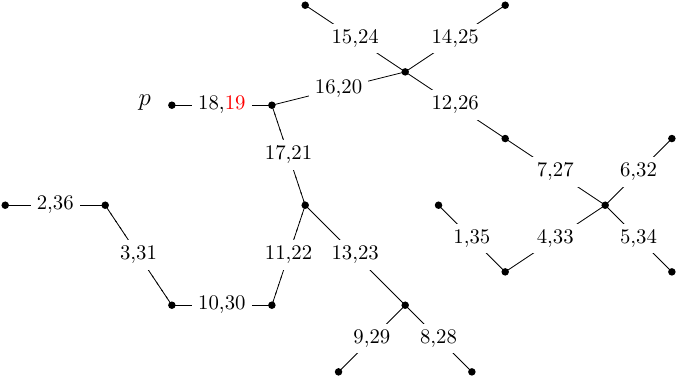}
			\caption{Adding the second breadth-first search labelling.}
		\end{subfigure}
		\caption{The pivot labelling of an example tree graph. A first labelling converging to pivot vertex $p$ is shown, which is then complemented by a second broadcasting labelling from $p$. Label 19 (shown in red) is redundant and removed.
		}
		\label{fig:pivot_labelling}
	\end{figure}
	
	Note that the pivot labelling is proper since it is incremental (there is a skip concerning the removed label, but larger labels can just be shifted to obtain an incremental labelling). The resulting temporal graph 
	is temporally connected since by design all vertices can reach pivot vertex $p$ at time $n-1$, and starting at time $n-1$, vertex $p$ can reach all vertices. It is also a minimal labelling since removing any label $\leq n-1$ on a path from a leaf vertex $f$ to $p$ reduces the reachability of $f$, and removing any label $> n-1$ makes it so $f$ cannot be reached by some other leaf vertex.
	
	\begin{lemma}
		\label{lemma:trees_Tplus_geq_2n3}
		$T^+(\texttt{Trees}) \geq 2n-3$.
	\end{lemma}
	
	\begin{proof}
		Since $\ell = 2n-1$ at the end of the labelling, a total of $2n-2$ labels are assigned, of which one is removed.
	\end{proof}
	
	\begin{lemma}
		\label{lemma:trees_tauplus_geq_2}
		$\tau^+(\texttt{Trees}) \geq 2$.
	\end{lemma}
	
	\begin{proof}
		For all non-trivial tree graphs, the pivot labelling assigns 2 labels to some edge. 
	\end{proof}
	
	To prove \Cref{lemma:trees_Tplus_geq_2n3,lemma:trees_tauplus_geq_2} are tight, \textit{i.e.} no denser labellings exist, we first present a lemma focussing on bridge edges, which is then applied on all the edges of tree graphs. 
	
	\begin{lemma}
		\label{lemma:bridge_edge}
		For any bridge edge $e$ of graph $G$ and any minimal labelling $\lambda$ such that $\mathcal{G} = (G, \lambda)$ is temporally connected, $\lambda$ can assign at most two labels to $e$.
	\end{lemma}
	
	\begin{proof}
		See also \Cref{fig:bridge_edge}. Consider a temporally connected graph $\mathcal{G}$ with bridge edge $e = \{u,v\}$, separating $\mathcal{G}$ into two temporal subgraphs $\mathcal{G}_1$ (with vertex $u$) and $\mathcal{G}_2$ (with $v$). 
		Suppose by contradiction that the labelling $\lambda$ of $\mathcal{G}$ assigns more than two labels to edge $e$, say $k > 2$ labels $\ell_1, \ell_2, ..., \ell_k$, and that this labelling is minimal.
		Define $t^-_u$ to be the earliest time at which all vertices in $\mathcal{G}_1$ are able to reach $u$. Similarly, define $t^+_v$ to be the latest time at which all vertices in $\mathcal{G}_2$ can be reached by $v$. Since $\mathcal{G}$ is temporally connected, there exists some label $\ell_i$ of $e$ such that $t^-_u < \ell_i < t^+_v$. Keeping $\ell_i$ is thus sufficient for maintaining reachability from all vertices in $\mathcal{G}_1$ to all vertices in $\mathcal{G}_2$.
		A symmetrical argument can be used to find label $\ell_j$ which is sufficient for maintaining reachability from all vertices in $\mathcal{G}_2$ to all vertices in $\mathcal{G}_1$. Together, $\ell_i$ and $\ell_j$ are thus sufficient for reachability concerning journeys using edge $e$, and edge $e$ can trivially be ignored for reachability between vertices in $\mathcal{G}_1$ (\textit{resp.} $\mathcal{G}_2$). 
		This results in all other labels on edge $e$ being redundant, which contradicts our assumption of $\lambda$ being a minimal labelling.
	\end{proof}
	
	\begin{figure}[h]
		\centering
		\includegraphics[width=.7\textwidth]{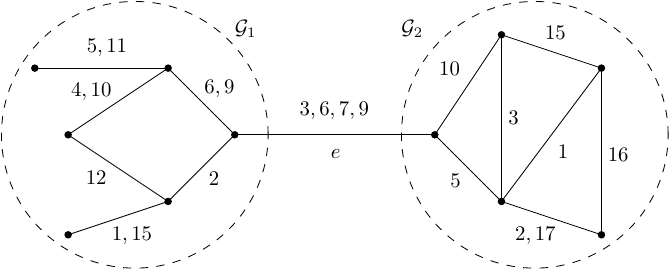}
		\caption{Example temporal graph with bridge edge $e$ dividing graph into two temporal subgraphs $\mathcal{G}_1$ and $\mathcal{G}_2$. Here we have $t_u^- = 6$, $t_v^+ = 10$, and $t_v^- = 5$, $t_u^+ = 8$. All labels of $e$ except for some label $6 < l_i < 10$ and some label $5 < l_j < 8$ can thus be removed from edge $e$. 
			Note that $l_i$ and $l_j$ may both be label 7 in this example. 
		}
		\label{fig:bridge_edge}
	\end{figure}
	
	\begin{corollary}
		\label{corollary:trees_tauplus_leq_2}
		$\tau^+(\texttt{Trees}) \leq 2$ .
	\end{corollary}
	
	\begin{proof}
		Tree graphs contain only bridge edges, so by \Cref{lemma:bridge_edge} no edge of a tree graph may have more than two labels in a minimal labelling.
	\end{proof}
	
	\begin{theorem}
		\label{theorem:trees_tauplus}
		$\tau^+(\texttt{Trees}) = 2$.
	\end{theorem}
	
	\begin{proof}
		\Cref{lemma:trees_tauplus_geq_2} and \Cref{corollary:trees_tauplus_leq_2}.
	\end{proof}
	
	We remark that \Cref{corollary:trees_tauplus_leq_2} implies that $T^+(\texttt{Trees}) \leq 2n-2$, which is only off by 1 from the lower bound of \Cref{lemma:trees_Tplus_geq_2n3}. We finish this section with proving the latter to be tight, which is less trivial than for $\tau^+$.

	We first need the following lemma concerning minimal temporally connected path graphs.
	
	\begin{lemma}
		\label{lemma:paths_Tplus_eq_2n3}
		$T^+(\texttt{Paths}) = 2n-3$
	\end{lemma}

	\newcommand{\Jonetwo}{\ensuremath{\mathcal{J}_{1,2}}\xspace}
	\newcommand{\Jtwoone}{\ensuremath{\mathcal{J}_{2,1}}\xspace}
	\newcommand{\liright}{\ensuremath{\ell_i^\rightarrow}\xspace}
	\newcommand{\liplusright}{\ensuremath{\ell_{i+1}^\rightarrow}\xspace}
	\newcommand{\liminusright}{\ensuremath{\ell_{i-1}^\rightarrow}\xspace}
	\newcommand{\lileft}{\ensuremath{\ell_i^\leftarrow}\xspace}
	\newcommand{\liplusleft}{\ensuremath{\ell_{i+1}^\leftarrow}\xspace}
	\newcommand{\liminusleft}{\ensuremath{\ell_{i-1}^\leftarrow}\xspace}
	\begin{proof}			
		Since path graphs are tree graphs, we have through \Cref{lemma:trees_Tplus_geq_2n3} and \Cref{corollary:trees_tauplus_leq_2} that  $2n-3 \leq T^+(\texttt{Paths}) \leq 2n-2$. 
		
		Note that for a path graph to be temporally connected, we only need to ensure that both extremities, $v_1$ and $v_n$, can reach each other, via a journey in one direction, and a journey in the other. Indeed, any other pair of vertices can use these journeys to reach each other. 
		Thanks to this, we can reason on temporally connected path graphs and only need to worry about the reachability between these two leaf vertices, instead of between all vertices.
		
		Suppose by contradiction that a minimal temporally connected path graph $G$ exists with $2n-2$ labels.
		By our previous argument, we have that any label which is not part of either the journey from $v_1$ to $v_n$, or of the journey from $v_n$ to $v_1$, is redundant.
		Each of these two journeys is composed of $n-1$ labels, meaning that to obtain $2n-2$ necessary labels from only these two journeys, we must have that all labels on the two journeys must be distinct.
		There must exist one edge $e_i$ such that the corresponding labels have the smallest difference among all edges of path $G$. By the temporal nature of journeys, there cannot be more than one such edge as the difference must necessarily increase on edges further away from $e_i$.
		Consider the labels on edge $e_i$ and incident edges $e_{i-1}$ and $e_{i+1}$, denoted $\ell_i^\rightarrow$, $\ell_i^\leftarrow$, $\ell_{i-1}^\rightarrow$, $\ell_{i-1}^\leftarrow$, $\ell_{i+1}^\rightarrow$, and $\ell_{i+1}^\leftarrow$. 
		(If edge $e_i$ only has one incident edge, then ignore the following concerning the non-existent other edge and labels.)
		Consider labels $\ell_{i-1}^\rightarrow < \ell_i^\rightarrow < \ell_{i+1}^\rightarrow$ to be part of $v_1 \leadsto v_n$, and $\ell_{i-1}^\leftarrow > \ell_i^\leftarrow > \ell_{i+1}^\leftarrow$ to be part of $v_n \leadsto v_1$.
		Suppose \textit{w.l.o.g.} that $\ell_i^\rightarrow > \ell_i^\leftarrow$. Thus we have that $\liplusright > \liright > \lileft$ and also $\liplusleft < \lileft < \liright$. On edge $e_{i-1}$ however, two cases are possible:
		
		\begin{itemize}
			\item $\liminusright < \lileft$: this means label \liright is redundant as $\liminusright < \lileft < \liplusright$;
			\item $\liminusleft > \liright$: this means label \lileft is redundant as $\liplusleft < \liright < \liminusleft$.
		\end{itemize}
		
		Due to the inequalities presented, and the fact that the difference between \liminusleft and \liminusright must be larger than the difference between \lileft and \liright, at least one of the previous cases must be present, meaning at least one label must be redundant which is our contradiction and finishes the proof.
	\end{proof}
	
	Intuitively, this proof by contradiction may be generalizable for trees as follows: prove that on the path between some two leaf vertices there must be a label which is redundant regarding the reachability between these leaf vertices, which must thus be necessary for some other pair of leaf vertices. Consider the path between these latter leaf vertices and repeat. The contradiction should occur after no more other pairs of leaf vertices can be considered. Nevertheless, the proof in its current state is sufficient (and much less complicated) to prove the following.

	\begin{theorem}
		\label{theorem:trees_Tplus_eq_2n3}
		$T^+(\texttt{Trees}) = 2n-3$.
	\end{theorem}
	
	\begin{proof}
		\Cref{lemma:trees_Tplus_geq_2n3} states the lower bound and since path graphs are tree graphs, \Cref{lemma:paths_Tplus_eq_2n3} implies the upper bound.
	\end{proof}
	
	Due to the upper bound coming from the subclass \texttt{Paths}, this means that for the class of tree graphs excluding path graphs, \textit{i.e.} \texttt{Trees} $\setminus$ \texttt{Paths}, it may be possible to obtain denser results. This is disproved if the proof of \Cref{lemma:paths_Tplus_eq_2n3} is in fact generalizable for trees as discussed.
	
	As a side note, we remark that the maximum temporal cost (\textit{resp.} temporality) of tree graphs corresponds to the minimum temporal cost (\textit{resp.} temporality) of tree graphs. In other words, all minimal labellings of tree graphs contain exactly $2n-3$ labels and exactly 2 labels on all edges except one. For proofs of these minimisation costs, we refer the reader to \cite{zhang_broadcasting_2013}.
	
	For trees, the results are mixed: on one hand we exactly determined both maximum temporal cost and maximum temporality of tree graphs $T^+(\texttt{Trees})$ and $\tau^+(\texttt{Trees})$, but on the other hand both are very sparse and do not improve upon lower bounds of maximum temporal cost $T^+$ and maximum temporality $\tau^+$.
	
	\section{Asymptotically optimal lower bounds}
	\label{sec:better_lower_bounds}
	
	From related work in \Cref{sec:preliminaries} and results on trees in \Cref{sec:trees}, one may wonder if it is even possible to obtain a maximum temporality larger than two, let alone a linear-size one. Also, regarding the maximum temporal cost, from Axiotis and Fotakis' result it is unclear whether a quadratic-size temporal cost is possible in sparse graph classes. 
	In this section, we show all of this is indeed possible, through two constructions.
	
	\begin{definition}[Ad-hoc construction] 
		\label{def:adhoc_construction}
		See also \Cref{fig:adhoc_construction}.
		Let $V$ be the set of vertices $a$, $b$, $k-1$ triplets $u_i$, $v_i$, $w_i$ with $i\leq k-1$, and vertices $u_k$ and $v_k$. The size of $V$ is thus $n = 3k + 1$. 
		Add edges $\{u_i, a\}$ with label $ik^2 -1$ for all vertices $u_i$, edge $e = \{a, b\}$ with labels $1$ and $ik^2$ for all natural numbers $i \leq k$, and edges $\{b, v_i\}$ with label $ik^2 + 1$ for all vertices $v_i$.
		Then, add edges $\{v_i, w_i\}$ with label $k^4$ and edges $\{u_i, u_{i+1}\}$ with label $k^4 + k-i$, as well as edges $\{v_i, v_{i+1}\}$ with label $k-i$, for all natural numbers $i \leq k-1$. 
		For all vertices $u_i$ and $w_j$ such that $i > j$, add edge $\{u_i, w_j\}$ with label $ik^2-1-j$.
	\end{definition}
	
	\begin{figure}[h]
		\begin{subfigure}{.5\textwidth}
			\includegraphics[width=\textwidth]{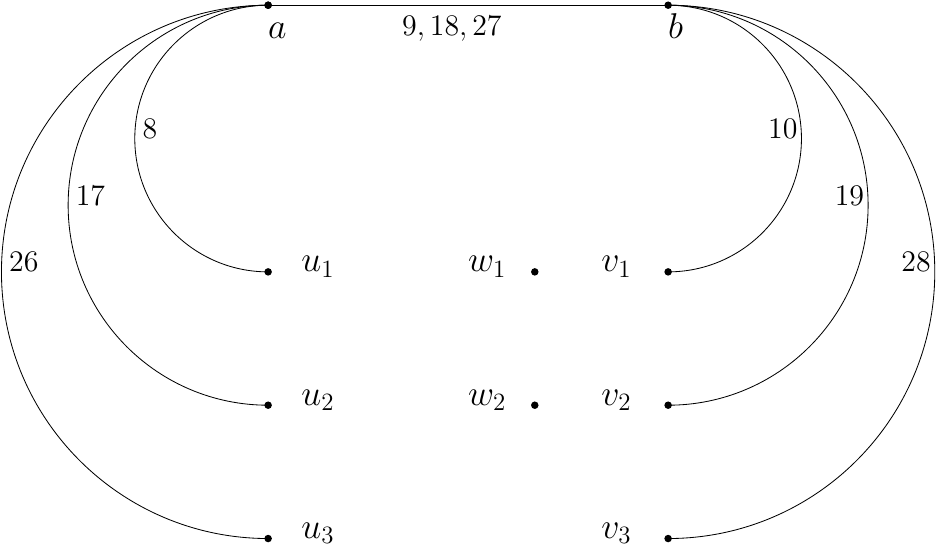}
			\caption{Initial step}
		\end{subfigure}
		\hfill
		\begin{subfigure}{.5\textwidth}
			\includegraphics[width=\textwidth]{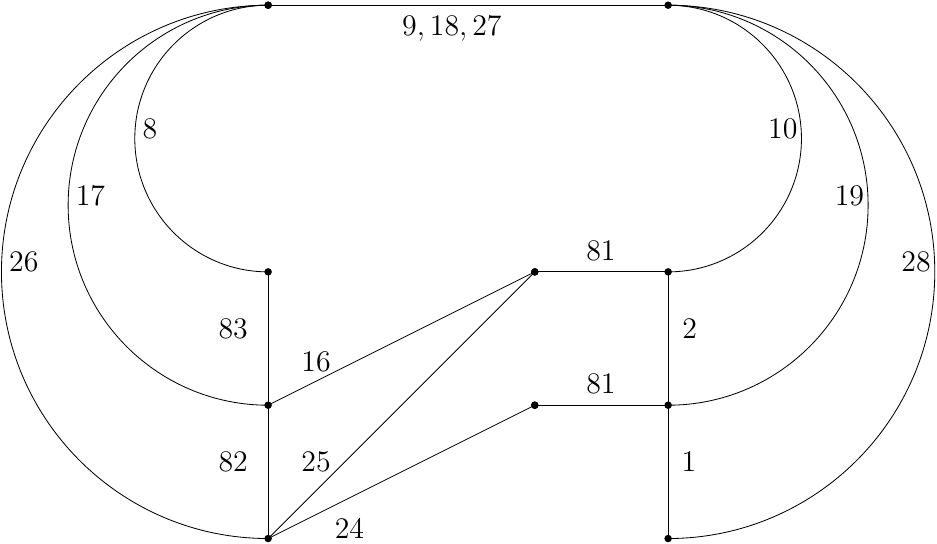}
			\caption{Completed temporally connected version}
		\end{subfigure}
		\caption{The ad-hoc construction for $k=3$, initialised with three necessary labels on edge $e$, and completed to a minimal and temporally connected graph. 
		}
		\label{fig:adhoc_construction}
	\end{figure}
	
	\begin{lemma}
		\label{lemma:adhoc_minimal}
		The ad-hoc construction yields a minimal temporally connected graph.
	\end{lemma}
	
	\begin{proof}
		Since all edges but $e$ have only one label, we will identify these labels with their corresponding edges for simplicity in the following proof.
		
		Let us first show the ad-hoc construction is temporally connected, by going over all vertices and showing they can reach all others. All vertices can trivially reach their neighbours, so we only look at the vertices which aren't neighbours.
		
		\begin{itemize}
			\item Vertex $a$ can reach vertices $v_i$ and $w_i$ by journey $(a, b, v_i, w_i)$ using label $k^2$ on edge $e$. 
			\item Vertex $b$ can reach vertices $w_i$ through journey $(b, v_i, w_i)$, and vertices $u_i$ through journey $(b, a, u_k, u_{k-1}, u_{k-2}, ..., u_i)$ using label $k^2$ on edge $e$. 
			\item Vertices $u_i$ can reach vertex $b$ and vertices $v_j$ and $w_j$ such that $i \leq j$ through journey $(u_i, a, b, v_j, w_j)$ using label $ik^2$ on edge $e$. They can reach the other vertices $v_j$ and $w_j$ such that $i > j$ by journey $(u_i, w_j, v_j)$. Finally, vertices $u_i$ can reach all vertices $u_j$ such that $i > j$ through journey $(u_i, u_{i-1}, u_{i-2}, ..., u_j)$, and vertices $u_j$ such that $i<j$ through journey $(u_i, a, u_j)$.
			\item Vertices $v_i$ can reach vertex $a$ through journey $(v_i, v_{i-1}, v_{i-2}, ..., v_1, b, a)$ using label $2k^2$, and from there on reach vertices $u_j$ through journey $\{a, u_k, u_{k-1}, u_{k-2}, ..., u_j\}$. 
			Vertices $v_i$ can reach vertices $v_j$ and $w_j$ such that $i \geq j$ through journey $(v_i, v_{i-1}, v_{i-2}, ..., v_j, w_j)$, and can reach vertices $v_j$ and $w_j$ such that $i < j$ through journey $(v_i, b, v_j, w_j)$.
			\item Finally, vertices $w_i$ can reach vertices $u_j$ through journey $(w_i, u_k, u_{k-1}, u_{k-2}, ..., u_j)$.
			They can reach vertices $a$ and $b$ through journey $(w_i, u_{i+1}, a, b)$ using label $(i+1)k^2$ on edge $e$, and from there on reach vertices $v_j$ and $w_j$ such that $i < j$ through journey $(a, v_j, w_j)$. Finally, they can reach vertices $v_j$ and $w_j$ such that $i > j$ through journey $(w_i, u_k, w_j, v_j)$.
		\end{itemize}
		Since all vertices can reach each other, the temporal graph is temporally connected.
		Let us finish by proving this labelling is minimal as well, by going over all labels and showing they are necessary.
		
		\begin{itemize}
			\item Labels $ik^2$ on edge $e$ are necessary for vertex $u_i$ to reach vertex $v_i$. 
			\item Edges $\{u_i, a\}$ are necessary for $u_i$ to reach $v_i$.
			\item Edges $\{b, v_i\}$ are necessary for $b$ to reach $v_i$.
			\item Edges $\{u_i, u_j\}$ are necessary for $u_k$ to reach $u_1$.
			\item Edges $\{v_i, v_j\}$ are necessary for $v_k$ to reach $v_1$.
			\item Edges $\{v_i, w_i\}$ are necessary for $v_i$ to reach $w_i$.
			\item Edges $\{u_i, w_j\}$ are necessary for $u_i$ to reach $w_j$.
		\end{itemize}
	\end{proof}
	
	%
	
	\begin{theorem}
		\label{theorem:adhoc_construction_tauplus}
		$\tau^+ \geq \tfrac{1}{3}n - O(1)$.
	\end{theorem}
	
	\begin{proof}
		By \Cref{lemma:adhoc_minimal}, the ad-hoc construction yields a minimal temporally connected graph, which has temporality $k = \tfrac{1}{3}n - \tfrac{2}{3}$ (the number of labels on edge $e$). 
	\end{proof}
	
	The temporality obtained through the ad-hoc construction is linear in $n$, which by \Cref{theorem:tauplus_upperbound} is asymptotically optimal.
	As a bonus, the ad-hoc construction also admits a large temporal cost, albeit comparable to the dense construction from Axiotis and Fotakis \cite{axiotis2016size}.
	
	\begin{corollary}
		\label{corollary:adhoc_construction_Tplus}
		$T^+ \geq \tfrac{1}{18}n^2 + \tfrac{31}{18}n - O(1)$.
	\end{corollary}
	
	\begin{proof}
		The ad-hoc construction was proven minimal and temporally connected in \Cref{lemma:adhoc_minimal}. Counting the total number of labels yields:
		\begin{align*} 
			T&=k + k + k + k-1 + k-1 + k-1 + (1 + 2 + 3 + ... + k-1)\\
			&= 6k -3 + \frac{k(k-1)}{2}\\
			&= \frac{1}{2}k^2 + \frac{11}{2}k -3\\
			&= \frac{1}{2}(\frac{1}{3}n - \frac{1}{3})^2 + \frac{11}{2}(\frac{1}{3}n - \frac{1}{3}) -3\\
			&= \frac{1}{18}n^2 + \frac{31}{18}n - \frac{43}{9}
		\end{align*}
		The dense part is situated between vertices $v_i$ and $w_j$.
	\end{proof}
	
	The ad-hoc construction is asymptotically as dense Axiotis and Fotakis, even sharing the same quadratic factor of $\tfrac{1}{18}$, but ours is slightly denser as their linear term is at most $\tfrac{3}{2}n$. 
	
	Since tree graphs resulted in small maximum temporality and temporal cost, and since the ad-hoc construction, as well as Axiotis and Fotakis' result, is built using a dense underlying graph, evidence may suggest that to obtain large maximum temporal cost and temporality, the underlying graph needs to be dense. We show this is not the case through the parity construction, which is a dense labelling of the sparsest class of graphs outside of trees: cycle graphs. This labelling was discussed
	during the open problem session of the Dagstuhl seminar on temporal graphs in 2021, after a preliminary presentation of this work (see \cite{casteigts_et_al:DagRep.11.3.16} for the report).
	
	\begin{definition}[Parity construction] See also \Cref{fig:parity_construction}.
		Let $n$ be even. Let $V$ be the set of vertices $v_i$ with $0 \leq i \leq n-1$, and let the set $E$ contain edges $\{v_i, v_{(i+1)\%n}\}$ for all $i < n$. Graph $G = (V, E)$ is thus a cycle graph. Note that $G$ admits exactly two maximum matchings. Assign all $\lceil \tfrac{n}{4} \rceil$ odd labels between 1 and $\tfrac{n}{2}$ included to all edges of one of the matching, say the one containing edge $\{v_0, v_1\}$, and all $\lfloor \tfrac{n}{4} \rfloor$ even labels to all edges of the other matching.
	\end{definition}
	
	\begin{figure}[h]
		\begin{center}
			\includegraphics[width=.37\textwidth]{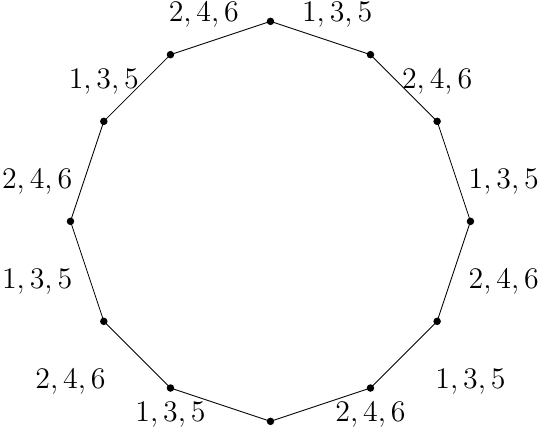}
		\end{center}
		\caption{The parity construction for $n=12$. Note that, even with the labelling, the graph is vertex-transitive, \textit{i.e.} no vertex is structurally different than any other vertex.
		}
		\label{fig:parity_construction}
	\end{figure}
	
	\begin{lemma}
		\label{lemma:parity_minimal}
		The parity construction yields a minimal temporally connected graph.
	\end{lemma}
	
	\begin{proof}
		The obtained graph, even with the labelling, is vertex-transitive. This means that to prove temporal connectivity, we only need to prove that one vertex can reach all others. Let us prove that vertex $v_0$ can reach all vertices. Following the journey starting with edge $\{v_0, v_1\}$ and label 1, by definition we have that this journey can continue in this direction, say clockwise, around the cycle using label $i+1$ after having used label $i$, and this up to label $\tfrac{n}{2}$ included. This implies that through this journey, vertex $v_n$ can reach all vertices $v_i$ such that $i \leq \tfrac{n}{2}$.
		Concerning the other vertices, consider the journey starting at $v_0$ using edge $\{v_{n-1}, v_0\}$ with label 2, going around the cycle in a counter-clockwise manner, using incrementing labels as well. This counter-clockwise journey reaches all vertices of $V$ not reached by the clockwise journey, and vice-versa.
		
		Concerning minimality: any label $i$ on any edge $\{v_j, v_{(j+1)\%n}\}$ is necessary for vertex $v_a$ to reach vertex $v_b$, with $a = (j-i+1) \% n$ and $b = (j+1+\tfrac{n}{2}-i) \% n$. 
	\end{proof}
	
	Although a fairly simple labelling of a simple graph, this labelling beats Axiotis and Fotakis' result by a large factor.
	
	\begin{theorem}
		$T^+ \geq T^+(\texttt{Even Cycles}) \geq \tfrac{1}{4}n^2$.
	\end{theorem}
	
	\begin{proof}
		By \Cref{lemma:parity_minimal}, the parity construction results in a minimal temporally connected graph. By definition, for $n$ a multiple of four, this construction assigns $\tfrac{n}{4}$ labels on each of the $n$ edges of the underlying cycle graph, giving a total of $\tfrac{1}{4}n^2$ labels. If $n$ is not a multiple of four, then each edge of one of the matchings is assigned $\lceil \tfrac{n}{4} \rceil = \tfrac{n}{4} + \tfrac{1}{2}$ labels, and each edge of the other matching contains $\lfloor \tfrac{n}{4} \rfloor = \tfrac{n}{4} - \tfrac{1}{2}$, for a total of $\tfrac{1}{4}n^2$ labels as well.
	\end{proof}
	
	Also, this construction results in a large temporality, although the ad-hoc construction admits a better one.
	
	\begin{corollary}
		$\tau^+(\texttt{Even Cycles}) \geq \lceil \tfrac{1}{4}n \rceil$.
	\end{corollary}
	
	
	Concluding this section, we thus already obtain asymptotically optimal results and beat previous work concerning maximum temporal cost and maximum temporality. For the former, the parity construction is the best result thus far, with $T^+ \geq \tfrac{1}{4} n^2$, and concerning the latter, the ad-hoc construction is the largest, with $\tau^+ \geq \tfrac{1}{3} n - O(1)$. In the next section, we decide to focus on another simple graph class, one which we now know does admit dense labellings, being cycle graphs. We end up beating both constructions from this section with one labelling. Nevertheless, we argue that the constructions in this section remain interesting: the parity construction for its simplicity and symmetry; and the ad-hoc construction serves as a basis for a result in \Cref{sec:conclusion} concerning other types of labellings.
	
	\section{Cycle graphs}
	\label{sec:cycles}
	
	After \Cref{sec:better_lower_bounds}, we focus on finding dense labellings of cycle graphs, which we know exist thanks to the parity construction. For this, we decide to adapt a temporal graph generator, which ultimately leads us to a the densest labelling in this paper, the generator labelling.
	
	\subsection{Adaptation of STGen}
	
	STGen is a generator of temporal graphs, made specifically for happy labellings (called ``simple'' labellings at the time of creation) by Arnaud Casteigts. It efficiently generates all happy temporal graphs of some given size which is useful when wanting to test some temporal graph property or find a counterexample. As an example, in \cite{casteigts2021temporal} STGen was used to verify some temporal spanner property on all temporal cliques up to size $n=7$.
	A brief and incomplete description of STGen follows. For more information and details, see \cite{casteigts2020efficient}. 
	The main point for the generator to be able to generate all (infinitely many) happy labellings, is that of bijective-equivalence group representatives. These are labellings with the smallest labels possible among their group of bijective-equivalent graphs, and more importantly, there is a finite amount of them given graph order $n$. These are in fact the labellings that are generated by STGen. In other words, a specific happy labelling may never actually be generated, but its representative will, which is sufficient when working on some connectivity property.
	STGen starts of with a complete graph and an empty labelling, which will be the root of the generation tree. Then, he creates a ``child'' temporal graph for every possible matching and assigns label 1 on these edges. Then, for every child, create their children by considering all possible matchings among unlabelled edges and assigning label 2 to them. This process is repeated recursively with incrementing labels until no more unlabelled labels can be assigned. The nodes of the generation tree now contain all bijective-equivalent representatives of happy labellings. Lots of clever optimisation techniques are applied during this process, from colouring lemmas, temporal isomorphism and automorphism, to rigidity and dissimilarity inheritance.
	
	We adapt STGen in the following manners, with some details and reasoning for each point below: 
	\begin{itemize}
		\item initialize with a cycle graph;
		\item cut branch if labelling is non-minimal (and discard);
		\item cut branch if labelling is temporally connected (and store);
		\item already labelled edges are still considered in descendants matchings;
		\item the stop condition, which activates when no unlabelled edges remain, is removed.
	\end{itemize}
	
	Let us refer to STGen with these adaptations as PTGen.
	
	The first point not only gives the benefit of working with only a sparse amount of edges to choose from for the matchings of children, but also has the nice property of encoding the reachability of vertices in the temporal graph as a continuous part of the cycle graph, which can be represented by the two extreme vertices only. This in turn allows for important operations such as a comparison or a union of reachabilities to be computable in constant time.
	
	For the second and third point, efficient testing of minimality and temporal connectivity is essential. The naive manner would be to test independently each child $\mathcal{G}$, taking time $O(n(m \log L + n \log n))$ to create the corresponding reachability graph $R(\mathcal{G})$ to check for temporal connectivity using the algorithm from \cite{xuan2003computing}, where $L$ is the largest label and $m$ the number of (labelled) edges, and repeating that $O(nL)$ times to check the necessity of label $\ell$ for each of the $O(nL)$ labels by removing it and constructing from scratch the reachability graph $R(\mathcal{G}^{-\ell})$.
	We present a more efficient method by maintaining accessibility variables $A(u)$ for all vertices $u$, which contain all vertices which can reach vertex $u$,
	and variables $A^{-\ell}(u)$ which contain all vertices which can reach vertex $u$ without using label $\ell$ (as with notation $G^{-\ell}$, here $\ell$ represents a single label on a single edge). 
	
	\begin{lemma}
		Updating the accessibility variables and testing for temporal connectivity and minimality, takes time $O(nT)$, where $T$ is the number of labels.
	\end{lemma}
	
	\begin{proof}
		Initially, variables $A(u)$ contain only vertex $u$. Then, whenever a matching $M$ of edges with label $\ell$ is added to the temporal graph, we initialize for all edges $e \in M$ variables $A^{-\ell, e}(u)$ as $A(u)$ for all vertices $u$.
		Now, for all edges $e = \{u,v\}$ with label $\ell$, update $A(u)$ and $A(v)$ to be $A(u) \cup A(v)$ which is done in constant time.
		Finally, for all edges $e = \{u, v\}$ with a label of value $\ell$ and all labels $\ell'$ with value less than $\ell$, update $A^{-\ell', e}(u)$ and $A^{-\ell', e}(v)$ to be $A^{-\ell', e}(u) \cup A^{-\ell', e}(v)$.
		These operations can be done in time $O(n^2 + n + nT) = O(nT)$.
		Testing for temporal connectivity can now be done in time $O(n)$ by verifying if all $A(u) = V$, and checking for minimality is doable in time $O(nT)$ by verifying if there exists a label $\ell$ and edge $e$ for which $A^{-\ell, e}(u) = A(u)$ for all $u$.\footnote{In reality, we use counters for constant time testing but this explanation suffices for asymptotic analysis.} Both get absorbed in the asymptotic notation.
	\end{proof}
	
	For a better comparison between the naive method and the one presented, note that $T = O(nL)$, and thus we outperform the naive method by a factor of $O(m \log L + n \log n)$.
	
	The last two adaptations of STGen are put in place so as to enable the generator to produce proper (and in particular, non-happy) labellings. Indeed, not considering the already labelled edges for the next matching with the next label directly counters the generation of non-happy proper labellings, and the stop condition, being when no unlabelled edges exist, indirectly counters the generation as well.
	We note that the colouring lemma of STGen, which states that edges with label $t$ should be incident to edges with label $t-1$, is still in use in PTGen which in particular for proper labellings means that label $t+1$ cannot be put on edges already having label $t$. 
	Let us show that PTGen terminates, even without a stop condition, thanks to our other branch cutting conditions. 
	
	\begin{lemma}
		PTGen terminates.
	\end{lemma}
	
	\begin{proof}
		Every time a new matching with a new label $\ell$ is added to the labelling, either the labelling becomes non-minimal and the branch is cut, or the labelling stays minimal, meaning that the latest labels added are all necessary, implying that the reachability (or accessibility) of at least one vertex must have strictly increased through labels $\ell$. 
		Since the reachability of vertices is at most the whole vertex set, the latter option can occur at most a finite amount of times before the corresponding labelling is temporally connected, at which point the branch is cut as well.
		Thus, every branch is eventually cut, whether it be because of non-minimality or temporal connectivity, and PTGen terminates.
	\end{proof}
	
	Other various optimisations were added, such as multi-threading and variable stacks. PTGen is coded in Rust and is available at \url{https://gitlab.com/echrstnn/max-temporality}.
	
	The generalisation to proper labellings introduces an explosion of possibilities compared to happy labellings, but the restrictions more than make up of for that, especially the fact that we limit ourselves to cycle graphs. As a rough comparison, PTGen ran on a standard laptop for cycles up to size $n=14$, taking a computation time of at most a couple of hours. This gave us enough evidence and intuition for our so-called generator labelling which was the main goal for adapting STGen.  
	
	%

\subsection{Generator labelling}
\label{subsec:generator_labelling}

For simplicity, we will first focus on the generator labelling for cycles of size $n$ being a multiple a four, denoting this class as \texttt{4k Cycles}. After the definitions, proofs and analysis, we state how to straightforwardly generalize to even cycles, and then to odd cycles for which the generator labelling is slightly different. The densities are slightly different between these specific subclasses, but cover the whole class of cycles, allowing us to draw conclusions for densities of \texttt{Cycles}.

\begin{definition}[Generator labelling of even cycle graph $G$]
	See also \Cref{fig:generator_labelling}. Let list $L_1$ initially contain all odd natural numbers between 1 and $n-1$ included, and list $L_2$ contain all even natural numbers between 1 and $n-1$ included. Let lists $L'_1$ and $L'_2$ be initially empty. 
	Take some edge $e$ of $G$, and let $e_{cc} = e_c = e$ and apply the following repeating process. 
	Put all labels from list $L'_1$ on edges $e_c$ and $e_{cc}$, and then assign the smallest label of $L_1$ to $e_c$, the second smallest to $e_{cc}$, the third back to $e_c$, and so on, distributing the labels of $L_1$ to edges $e_c$ and $e_{cc}$ in this alternating manner.
	Finally, move the smallest label from $L_1$ to $L'_1$ and remove the largest label from $L_1$, and move edge $e_c$ to be the next clockwise edge and edge $e_{cc}$ to be the next counter-clockwise edge. 
	Repeat this process, by taking lists $L_2$ and $L'_2$ next, and then repeat it again using lists $L_1$ and $L'_1$ again, and so forth alternating between these pairs of lists. The last time this process is repeated is when $e_c = e_{cc}$ again, which assigns labels to the last unlabelled edge, denoted $e'$. 
	Ultimately, assign label $\ell + 2$ to $e'$, where $\ell$ is the largest label on $e'$.
\end{definition}	

\begin{figure}[h]
	\begin{subfigure}{.5\textwidth}
		\includegraphics[width=.88\textwidth]{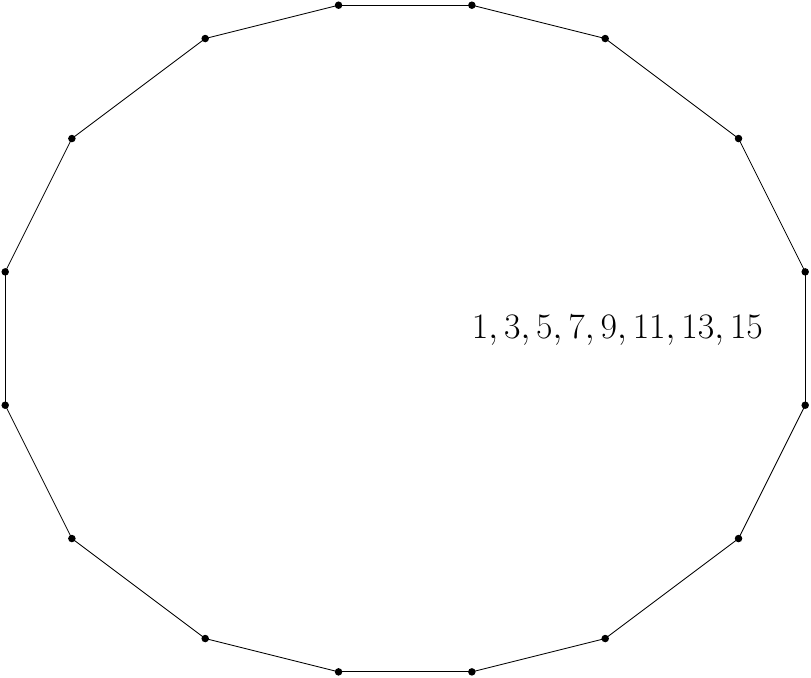}
		\caption{$e_c = e_{cc}$, $L_1 = (1, 3, 5, ..., 15)$ and $L'_1 = \emptyset$.}
	\end{subfigure}
	\hfill
	\begin{subfigure}{.5\textwidth}
		\includegraphics[width=.88\textwidth]{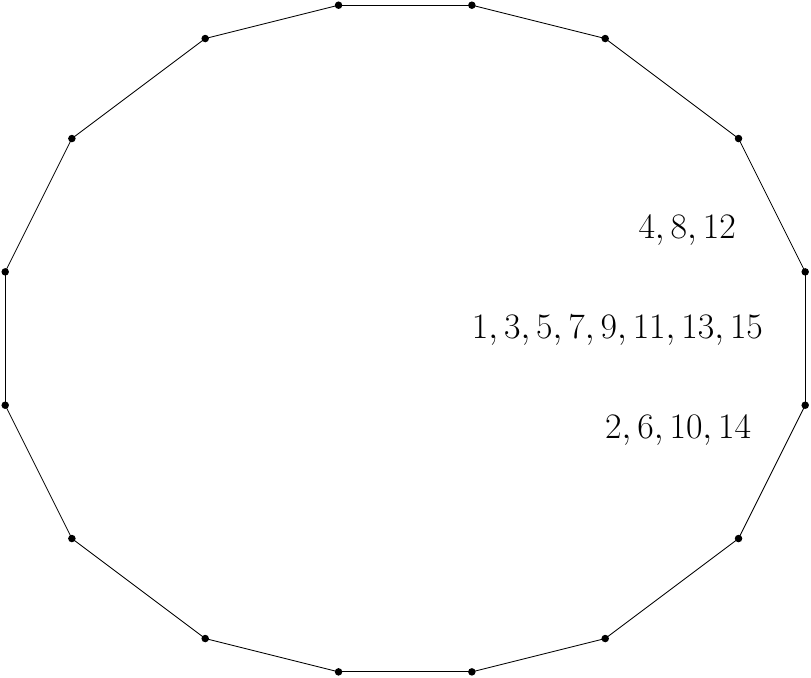}
		\caption{$L_2 = (2, 4, 6, ..., 14)$ and $L'_2 = \emptyset$.}
	\end{subfigure}
	\begin{subfigure}{.5\textwidth}
		\includegraphics[width=.88\textwidth]{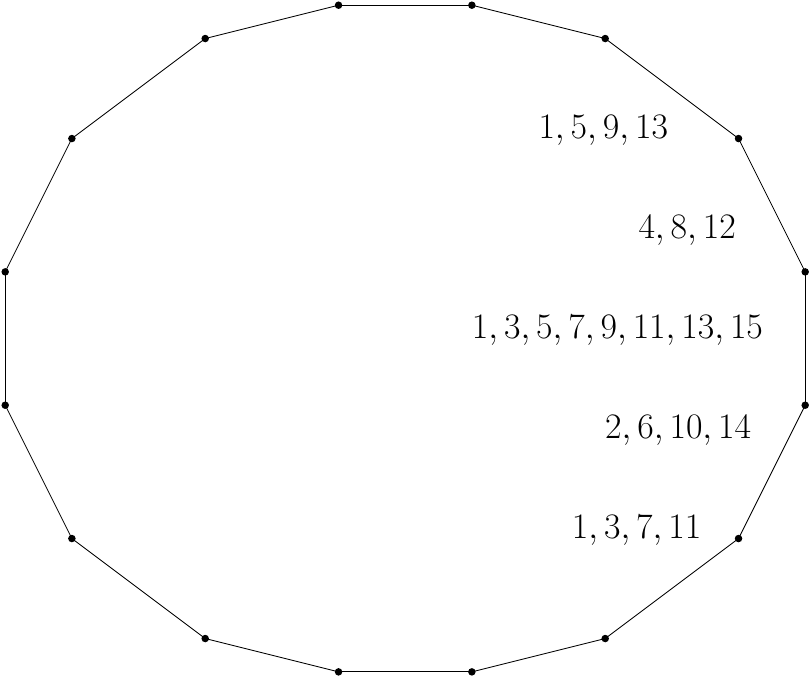}
		\caption{$L_1 = (3, 5, 7, 9, 13)$ and $L'_1 = (1)$.}
	\end{subfigure}
	\hfill
	\begin{subfigure}{.5\textwidth}
		\includegraphics[width=.88\textwidth]{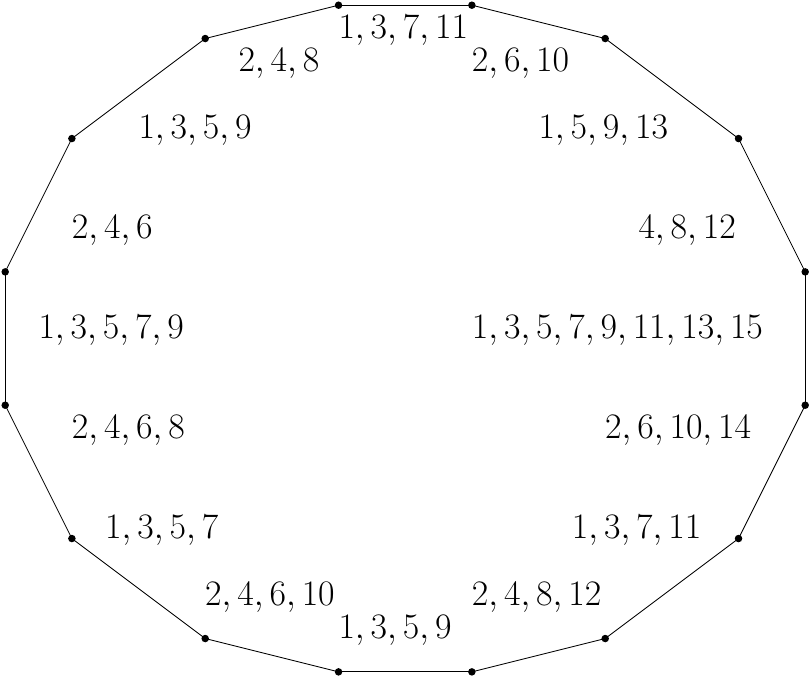}
		\caption{$L_1 = \emptyset$, $L'_1 = (1, 3, 5, 7)$ and $e_c = e_{cc}$ again.}
	\end{subfigure}
	\caption{The generator labelling for $n=16$, starting on rightmost edge $e$, and repeating process with lists in captions. After the processes finish repeating, label $9$ is assigned on leftmost edge $e'$. 
	}
	\label{fig:generator_labelling}
\end{figure}

Proving this labelling results in a minimal and temporally connected graph for any order $n$ cycle graph is more complex than for previous labellings or constructions. This is mostly due to the inherent unreadability of the multitude of journeys at play in the labelling representation of these temporal graphs. 
For this reason, we introduce a different type of representation which is very similar to the so-called link stream representation used in various temporal graph theory papers \cite{latapy2018stream, li2018streaming, zhao2016link}, and thus by slight misuse of terminology, we simply refer to it as the link stream representation.
Link streams intuitively focus more on the time edge aspect of a temporal graph, and less on the structure of the underlying graph. Since we already know the underlying structure of the graph (being a cycle graph), and since we have trouble distinguishing journeys with the labelling representation, link streams seem to be perfect to represent our generator labelling. As an illustrative example, the link stream representation of \Cref{fig:generator_labelling} is given in \Cref{fig:generator_link_stream}, where all edges of the cycle are represented in one dimension, the horizontal dimension, and time is represented in the other, the vertical dimension. A ``label'' is thus represented as a time edge at the intersection of the corresponding edge and value.

\begin{figure}[]
	\centering
	\begin{subfigure}{.75\textwidth}
		\includegraphics[width=\textwidth]{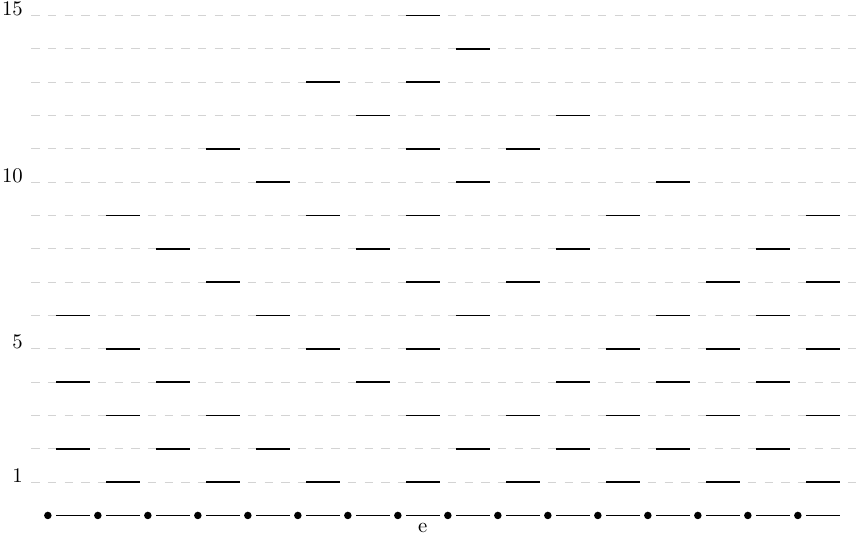}
		\caption{Without journeys.}
	\end{subfigure}
	\begin{subfigure}{.75\textwidth}
		\includegraphics[width=\textwidth]{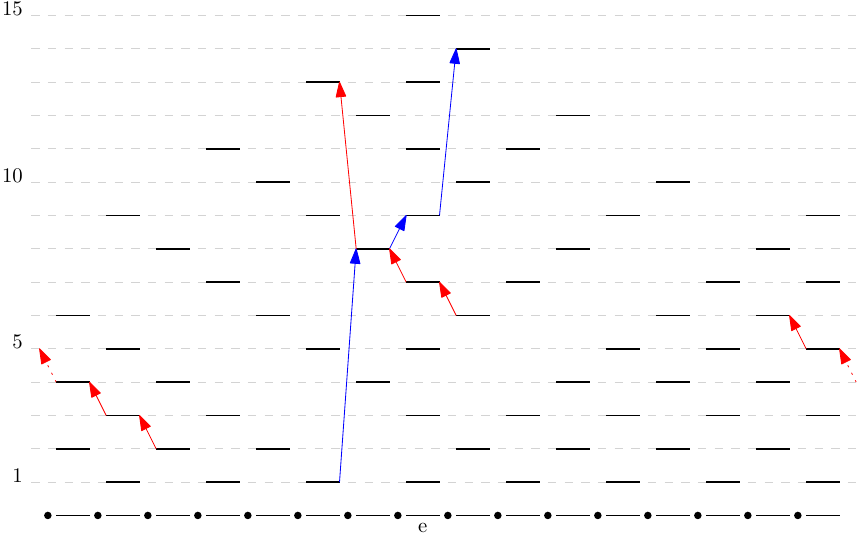}
		\caption{With two clockwise journeys in red and one counter-clockwise journey in blue. Only the outermost red journey is prefix-foremost.}
	\end{subfigure}
	\begin{subfigure}{.75\textwidth}
		\includegraphics[width=\textwidth]{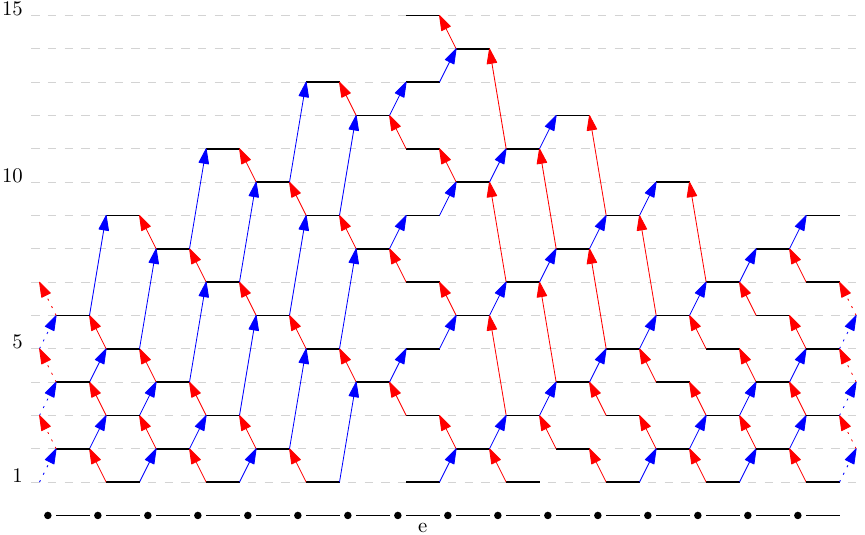}
		\caption{With all dominating journeys.}
	\end{subfigure}
	\caption{The generator labelling for $n=16$ in the link stream representation. Edge $e$ is shown in the middle, and edge $e'$ is completely to the right and connects the rightmost and leftmost vertices. 
	}
	\label{fig:generator_link_stream}
\end{figure}

In this representation, a journey informally corresponds to a (sometimes steep) ``staircase'' of time edges which, obeying the direction of time, cannot go down. 
We now remind and define some concepts concerning journeys in this setting, all of which are illustrated in \Cref{fig:generator_link_stream}. 
A prefix of a journey $(v_1, \ell_1, v_2, \ell_2, ..., \ell_{k-1}, v_k)$ is a part of the journey cut back from the arrival vertex, \textit{i.e.} $(v_1, \ell_1, v_2, \ell_2, ..., \ell_{i-1}, v_i)$ for some $i \leq k$, and a suffix of a journey is a part of the journey cut back from the starting vertex, \textit{i.e.} $(v_j, \ell_j, v_{j+1}, \ell_{j+1}, ..., \ell_{k-1}, v_k)$ again for some $j \leq k$.
A foremost journey from vertex $u$ to vertex $v$ is a journey which arrives at the earliest time among all journeys from $u$ to $v$. A prefix-foremost journey from $u$ to $v$ is a journey which prefixes are all foremost, In other words, a prefix-foremost journey always takes the earliest edges possible on its journey.
We define a clockwise journey as a journey which takes only time edges to the left, and a counter-clockwise journey is composed of only time-edges going to the right.
Finally, we define dominating journeys as clockwise (\textit{resp.} counter-clockwise) journeys such that no other clockwise (\textit{resp.} counter-clockwise) journey exists which covers all its vertices or more. In particular, since a clockwise (\textit{resp.} counter-clockwise) dominating journey is the only clockwise (\textit{resp.} counter-clockwise) journey covering its vertices, it is prefix-foremost.

The following technical lemmas are needed to then prove minimality of the generator labelling.

\begin{lemma}
	\label{lemma:dominating_journey_characterisation}
	In a cycle graph $G = (V, E)$ without any journey covering $V$, a clockwise journey $\mathcal{J} = ((\{v_j, v_{j-1}\}, t_1), (\{v_{j-1}, v_{j-2}\}, t_2), ..., (\{v_i, v_{i-1}\}, t_k))$ is dominating if and only if:
	\begin{itemize}
		\item it starts at the earliest date possible, \textit{i.e.} there exists no time edge $(\{v_j, v_{j-1}\}, t)$ nor $(\{v_{j+1}, v_{j}\}, t)$ with $t < t_1$;
		\item it ends at the latest date possible, \textit{i.e.} there exists no time edge $(\{v_i, v_{i-1}\}, t)$ nor $(\{v_{i-1}, v_{i-2}\}, t)$ with $t > t_k$;
		\item no other time edges exist between successive time edges, \textit{i.e.} for all successive pairs of time edges $(\{v_a, v_{a-1}\}, t_b)$ and $(\{v_{a-1}, v_{a-2}\}, t_c > t_b)$ of $\mathcal{J}$, there exists no time edge $(\{v_a, v_{a-1}\}, t)$ or $(\{v_{a-1}, v_{a-2}\}, t)$ with $t_b < t < t_c$.
	\end{itemize}
	A symmetric characterisation holds for counter-clockwise journeys.
\end{lemma}

\begin{proof}
	Let us focus on clockwise journeys, the proof being symmetric for counter-clockwise journeys. 
	Suppose by contradiction that a journey $\mathcal{J}$ obeys the three criteria, but is not dominating, meaning there exists some other distinct clockwise journey $\mathcal{J}'$ covering the same vertex set (or more). 
	A case analysis follows depending on which vertex $\mathcal{J}'$ starts.
	
	If journey $\mathcal{J}'$ starts from any vertex that $\mathcal{J}$ covers, except for $v_j$, then to ensure $\mathcal{J}'$ covers the vertices of $\mathcal{J}$, it needs to go all the way around the cycle graph and thus cover $V$, which is explicitly excluded in this lemma.  
	
	If $\mathcal{J}'$ starts at vertex $v_j$, and it contains some earlier time edge than the corresponding time edge in $\mathcal{J}$, then $\mathcal{J}$ doesn't respect criterion three (as this earlier time edge exists between successive time edges of $\mathcal{J}$). 
	If instead it contains a later time edge, then $\mathcal{J}'$ must rejoin or cross $\mathcal{J}$ at some point (since $\mathcal{J}$ uses the latest date of edge $\{v_i, v_{i-1}\}$ by criterion two), implying $\mathcal{J}$ again does not respect criterion three.
	Of course, if $\mathcal{J}'$ does not contain any earlier or later time edge than $\mathcal{J}$, then it will end in the same manner as $\mathcal{J}$ without any way of continuing by criterion two, meaning it is identical to $\mathcal{J}$. 
	
	Lastly, if $\mathcal{J}'$ starts on any other vertex, then since $\mathcal{J}$ respects criterion one, $\mathcal{J}'$ must arrive later than $t_1$ on edge $\{v_j, v_{j-1}\}$. By the same argument as before concerning $\mathcal{J}'$ having a later time edge than $\mathcal{J}$, the former must rejoin or cross the latter at some point, implying $\mathcal{J}$ does not respect criterion three.
	
	Since all cases end in some contradiction, being either $\mathcal{J}$ breaking one of the criteria or $\mathcal{J}'$ being identical to $\mathcal{J}$, we can thus conclude that $\mathcal{J}$ is dominating.
\end{proof}


%


%

\begin{lemma}
	\label{lemma:dominating_journeys_necessity}
	In a cycle graph $G = (V, E)$ without any journey covering $V$, a pair of clockwise and counter-clockwise journeys is necessary if:
	\begin{itemize}
		\item both start at some same vertex $v$;
		\item both are prefix-foremost;
		\item both are a suffix of a dominating journey;
		\item and they do not cross (except on vertex $v$).
	\end{itemize}
\end{lemma} 

\begin{proof}
	We prove that such a pair of clockwise and counter-clockwise journeys, say journey $\mathcal{J}_v^w$ which clockwise goes up to vertex $w$, and journey $\mathcal{J}_v^u$ which counter-clockwise goes up to vertex $u$, is necessary for reachability from $v$ to $w$ and from $v$ to $u$ respectively. \textit{W.l.o.g.} we give the proof for the former only, the proof for the latter being symmetric. We first prove that no counter-clockwise journey can reach vertex $w$, and then that the only clockwise journey that can reach $w$ is journey $\mathcal{J}_v^w$, from which it follows that this journey is necessary.
	
	First note that vertex $v$ cannot reach further than $u$ in a counter-clockwise manner. Indeed, if by contradiction we suppose there is some counter-clockwise journey $\mathcal{J}_v^{u'}$ from $v$ to vertex $u'$ such that $u'$ is positioned further than $u$, then \textit{w.l.o.g.} we may consider $\mathcal{J}_v^{u'}$ to be prefix-foremost (if it is not, then we can make it so by changing its time edges for the earliest possible).
	Since $\mathcal{J}_v^{u}$ and $\mathcal{J}_v^{u'}$ are both prefix-foremost clockwise journeys, 
	we know that $\mathcal{J}_v^{u}$ must be a prefix of journey $\mathcal{J}_v^{u'}$, \textit{i.e.} $\mathcal{J}_v^{u'}$ is the concatenation of journeys $\mathcal{J}_v^{u}$ and say $\mathcal{J}_u^{u'}$. Journey $\mathcal{J}_v^{u}$ is a suffix of a dominating journey $\mathcal{J}_d$, meaning no other counter-clockwise journey covers the vertices of $\mathcal{J}_d$ or more, but now we obtain our contradiction: the concatenation of $\mathcal{J}_d$ and $\mathcal{J}_u^{u'}$ covers more vertices (the only case where this wouldn't be true is if $\mathcal{J}_d$ covered $V$ which is explicitly excluded from the lemma statement). 
	Since $\mathcal{J}_v^w$ and $\mathcal{J}_v^u$ don't cross (except on vertex $v$), we now know that $v$ cannot reach $w$ through a counter-clockwise journey.
	
	To finish the proof, we show $\mathcal{J}_v^w$ is the only clockwise journey that can reach $w$, meaning all its edges are necessary. 
	Suppose by contradiction another clockwise journey $\mathcal{J}$ exists from $v$ to $w$. It cannot be a prefix-foremost journey, as by definition this would be journey $\mathcal{J}_v^w$. Since $\mathcal{J}$ is not prefix-foremost, it uses some edge $e$ with label $l'$ whereas $\mathcal{J}_v^w$ uses edge $e$ with some label $l < l'$. However, we remind the reader that $\mathcal{J}_v^w$ is a suffix of a dominating journey $\mathcal{J}_d$. Altogether, this means another journey exists covering the same vertices as $\mathcal{J}_d$, being the concatenation of the prefix of $\mathcal{J}_d$ up to vertex $v$, and $\mathcal{J}$. By definition of dominating journeys, this is a contradiction.
\end{proof}

We note that if the pair of journeys from \Cref{lemma:dominating_journeys_necessity} collectively covers $V$, then vertex $v$ can reach all vertices through these journeys.

%

\begin{theorem}
	\label{theorem:generator_labelling_minimalTC}
	The generator labelling yields a minimal temporally connected graph.
\end{theorem}

\begin{proof}
	The proof is by induction. 
	Consider cycle graph $C_8$ as our base case. Apply the generator labelling and compute the dominating journeys. See \Cref{fig:generator_labelling_proof}. Although we are considering cycle graphs in this section, the authors found that the stream link representation (very vaguely) resembles a tree, especially when considering larger cycles. For this reason, let us define some specific sets of time edges as follows. 
	Let the five earliest time edges on edges $\{v_{-2}, v_{-1}\}$, $\{v_{-1}, v_{1}\}$ and $\{v_{1}, v_{2}\}$ be referred to as the seed, and the three latest time edges as the trunk. Let the latest two time edges on edges $\{v_{-4}, v_{-3}\}$ and $\{v_{-3}, v_{-2}\}$ be referred to as a branch, as well as the ones on edges $\{v_{2}, v_{3}\}$ and $\{v_{3}, v_{4}\}$. More specifically, let the former be branch $B_2$, as the dominating clockwise journey starting at vertex $v_2$ ends in these edges, and the latter $B_{-2}$ as the dominating counter-clockwise journey starting at vertex $v_{-2}$ ends here. Finally, let the other time edges be referred to as the base. 
	Note that all time edges are part of some dominating journey, and that all dominating journeys start on a time edge with time 1 in the base (except for the two dominating journeys starting in the seed) and that all dominating journeys end in the latest time edges of branches (except for two dominating journeys ending in the trunk). 
	
	\begin{figure}[]
		\centering
		\begin{subfigure}{.75\textwidth}
			\includegraphics[width=\textwidth]{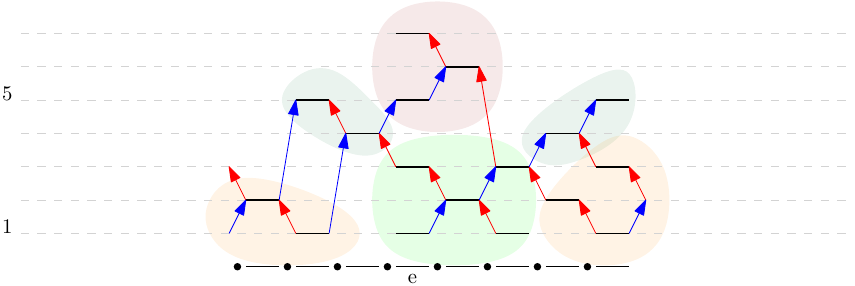}
			\caption{Base case cycle graph $C_8$, with seed, base, trunk and two branches.}
		\end{subfigure}
		\begin{subfigure}{.75\textwidth}
			\includegraphics[width=\textwidth]{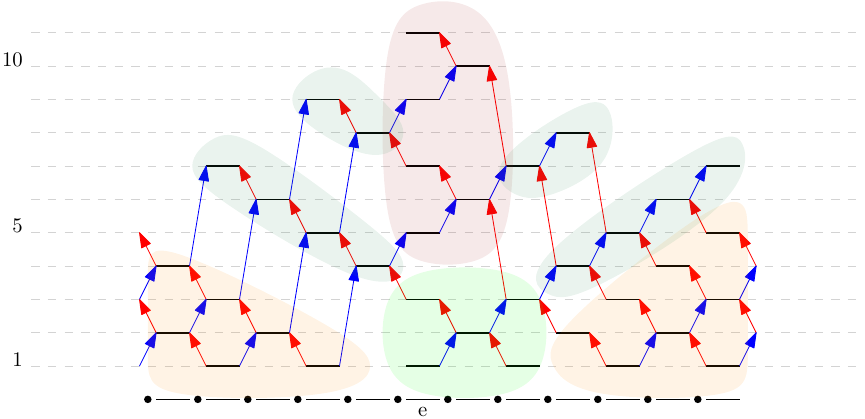}
			\caption{Induction step cycle graph $C_n$ (here $n=12$) before extending to $C_{n+4}$.}
		\end{subfigure}
		\begin{subfigure}{.75\textwidth}
			\includegraphics[width=\textwidth]{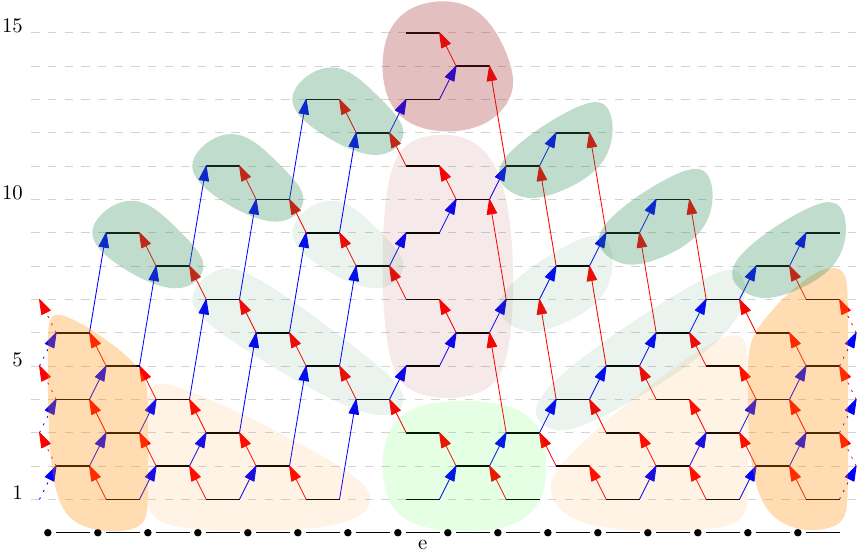}
			\caption{Apex, leaves, and roots extending $C_n$ to obtain $C_{n+4}$, adding two branches.}
		\end{subfigure}
		\caption{Illustration of the proof by induction for \Cref{theorem:generator_labelling_minimalTC}, with the seed (light green), trunk (brown), branches (green) and base (light brown). At the induction step, the apex, leaves, and roots are shown in the same colour as the structures they extend, but less transparent. 
		}
		\label{fig:generator_labelling_proof}
	\end{figure}
	
	It is possible to claim minimality and temporal connectivity for this small temporal graph, although we specifically point out that \Cref{lemma:dominating_journeys_necessity} can be used on all vertices (except for those of the seed) to prove necessity of all dominating journeys starting on these vertices, and reachability of all these vertices. 
	Note that now only the time edges of the base remain to be proven necessary. The time edges of the counter-clockwise dominating journey can be proven necessary by applying \Cref{lemma:dominating_journeys_necessity} on vertex $v_{-2}$, but it cannot be applied to vertex $v_2$ to prove the remaining time edges necessary, as its counter-clockwise prefix-foremost journey is not dominating. However, they are proven necessary through the ad hoc argument: without these edges $v_2$ cannot reach $v_{-4}$, as any clockwise journey by definition cannot reach it except for its clockwise dominating journey which relied on these time edges, and any counter-clockwise journey reaches at most vertex $v_4$. 
	Concerning reachability of the vertices of the seed, it is possible for them to use the dominating journeys starting in the seed to go to any other vertex (note that these journeys do not cross outside of in the seed, but do cover $V$). 
	
	Now, in the inductive step, this structure of seed, base, trunk and branches remains or gets extended when growing the generator labelling for some $C_n$ to some $C_{n+4}$. More precisely, the seed remains as is, the base gets extended with so-called roots, the trunk with a so-called apex, and the branches with leaves. Also, two new branches are created in every inductive step, which sprout from the top of the trunk.
	Underlying all this, we prove that in the inductive step, the dominating journeys get extended slightly, are modified, or created, in a precise manner which ultimately allows us to again use \Cref{lemma:dominating_journeys_necessity} to prove minimality and temporal connectivity, in a very similar manner as how we did for $C_8$.
	
	
	Now, suppose we have a cycle graph $C_{4k}$ with the generator labelling which has been proven minimal and temporally connected, specifically through applying  \Cref{lemma:dominating_journeys_necessity} on all vertices except for the seed.
	Add vertices $v_{2k + 1}$, $v_{2k + 2}$, $v_{-2k - 1}$, and $v_{-2k - 2}$ and the corresponding edges to the link stream representation so as to obtain cycle $C_{4k+4}$. See \Cref{fig:generator_labelling_proof}. This effectively breaks dominating journeys which previously used edge $\{v_{-2k}, v_{2k}\}$, whose time edges now belong to edge $\{v_{-2k}, v_{-2k-1}\}$. We will patch these halves of dominating journeys back together in what follows, although not exactly with their original half. 
	Note that now the generator labelling for $C_{4k+4}$ is exactly this labelling, with some additional time edges which are all later, \textit{i.e.} for all additional time edges $(e, t)$, there exists no already present time edge $(e, t' > t)$. 
	Let the three additional time edges extending the trunk be referred to as the apex, let the leaves be the pairs of additional time edges extending the branches (as well as creating branches $B_{2k}$ and $B_{-2k}$ from the trunk), and let the remaining additional time edges be the roots, which extend the base.
	
	Let us start by proving these additional edges are all part of some dominating journey.
	
	Leaves extending branches $B_i$ extend the corresponding dominating journey starting from vertex $v_i$. The three conditions from \Cref{lemma:dominating_journey_characterisation} hold for this extended journey, as it still starts at time 1, no additional time edges have been added in between the time edges it uses, and it ends at the latest time possible at the top of its respective leaves. Regarding the leaves that create a new branch $B_{2k}$ and $B_{-2k}$, these extend dominating journeys from vertices $v_{2k}$ and $v_{-2k}$ which ended at the top of the trunk before (we can observe two of these exist in $C_8$, and below we prove that in every inductive step two new such journeys are created). These extended journeys remain dominating by the same argument as for other branches.
	All leaves are thus part of a dominating journey.
	Note that this extends (by exactly two time edges) basically half of all previously existing dominating journeys.
	
	The other half of previous dominating journeys are broken up through the addition of the four new vertices and edges. 
	The roots serve to patch these journeys back together. Note that before, all these dominating journeys started in the base, cycled around, and climbed through the branches to finish at the top of some branch. More specifically, such a dominating journey starting from a vertex $v_i$ finished at the top of branch $B_{i-1}$ for $i > 0$ and at the top of branch $B_{i+1}$ otherwise (an exception being the journey starting from $v_{-3}$ which ends at the second largest time edge of the branch $B_{-2}$). We show this remains true after the inductive step. 
	The reconstructed dominating journey, suppose from vertex $v_i$ for $i>0$ (the explanation being symmetric for $i<0$), starts of with the same time edges it had before in the base until it reaches the roots. This means this part of the journey respects two of the conditions of \Cref{lemma:dominating_journey_characterisation}, being it starts at time 1, and no time edges exist in between its time edges as this journey was dominating before and the additional time edges are all later. Now the earliest four time edges possible are taken to continue this journey in the roots, cycling around to the other side of the link stream. This also respects the condition of having no time edges in between these four time edges, as the roots are densely packed by definition. The journey is now four time steps too late to reconnect with the other half it had before, connect it instead with the half of the journey which previously started from vertex $v_{i-4}$ for which it arrived through the roots perfectly on time. This latter half also respects the condition of having no time edges in between its time edges due to part of a dominating journey before, and no additional time edges have been added in between these time edges. Since our journey now follows the part of the dominating journey which previously started at $v_{i-4}$, it arrives at branch $B_{i-5}$, but can now continue through the leaves of $B_{i-3}$ and finally $B_{i-1}$ to end at the latest edge. By construction, this continuation through the leaves respects the conditions of \Cref{lemma:dominating_journey_characterisation} since there are no time edges in between, and it ends at the latest time possible.
	There are two exceptions to this: the reconstruction of the dominating journey starting from vertex $v_3$ only uses three time edges from the roots, before directly ending in the leaves of branch $B_2$, and the one starting from vertex $v_5$ directly goes up through the leaves of $B_2$ and $B_4$ after the roots. Both reconstructed journeys are dominating as well.
	Note that now all dominating journeys starting at vertices $v_i$ with $-2k \leq i \leq 2k$ have been extended (the ones cycling around have been broken apart and refitted first but in terms of length have been extended as well) by exactly two time edges compared to their previous length in $C_{4k}$. 
	
	Observe that some of the earliest roots have not been shown to be part of a dominating journey yet, and also that some halves of previous dominating journeys have not been refitted together yet. We show another four dominating journeys exist which start from the four new vertices, use these earliest roots, as well as the remaining parts of previous dominating journeys, and two of these journeys use the time edges of the apex. 
	Proving these four journeys are dominating is done through again applying \Cref{lemma:dominating_journey_characterisation} with the arguments already explained for the other dominating journeys, and thus we decide to forgo doing this again four more times.
	The dominating journey starting at $v_{-2k-1}$ goes clockwise, starting at time 1 and the four earliest roots, then it continues with the part of the previous dominating journey ending in branch $B_{-2k + 4}$, and goes up through the leaves of branches $B_{-2k + 2}$ and $B_{-2k}$, finishing at the latest time edge of the latter. 
	Starting at vertex $v_{-2k-2}$, we have a counter-clockwise dominating journey, starting at time 1 using only one root, linking up with part of a previous dominating journey which finished at branch $B_{2k-2}$, which is extended further through branch $B_{2k}$ and the apex to end on the second latest time edge. 
	We note that the last two dominating journeys do not cross, except on the first edge, and cover $V$.
	Continuing, we have a clockwise dominating journey starting at $v_{2k+2}$ and time 1, using two roots before using part of a previous dominating journey leading up to branch $B_{-2k-2}$, continuing through the leaves of $B_{-2k}$ and ending in the apex on the largest time edge.
	Finally, there's a counter-clockwise dominating journey from vertex $v_{2k+1}$ using three roots cycling around the link stream, pairing up with part of a previous dominating journey leading up to branch $B_{2k-4}$ which then continues through leaves of $B_{2k+2}$ and $B_{2k}$ to end on the largest time edge of that branch.
	Again, these two journeys do not cross, except on the first edge, and cover $V$.
	
	Thus, we have that all time edges are part of a dominating journey, and that basically the same journeys from $C_{4k}$ remain in $C_{4k+4}$ (albeit some of them recombined differently) and were extended by exactly two time edges. Since for all vertices but those of the seed, \Cref{lemma:dominating_journeys_necessity} was used to prove necessity of the corresponding dominating journeys, this lemma can be used again for these vertices to prove necessity of their corresponding dominating journeys, as well as reachability of these vertices to all others. For the time edges and vertices of the seed, the argument used for $C_8$ can be generalized to prove necessity and reachability as well. Finally, the last four dominating journeys which start on the four new vertices, can use \Cref{lemma:dominating_journeys_necessity} as well, since their clockwise and counter-clockwise prefix-foremost journeys are dominating and can collectively cover $V$. 
	
	In conclusion, we have proven that the base case, being the generator labelling for $C_8$, is minimal and temporally connected. Then, for any inductive step from $C_{4k}$ to $C_{4k+4}$, minimality and temporal connectivity are conserved in this labelling. Thus, the generator labelling produces a minimal and temporally connected graph for any size $4k$. (The generator labelling works for $C_4$ as well, it is easy to check, but the structure of the labelling was easier to explain with $C_8$, as for $C_4$ the labelling is composed of only the seed.)
\end{proof}


Let us now analyse the density of the generator labelling.

\begin{theorem}
	\label{theorem:generator_labelling_Tplus}
	$T^+ \geq T^+(\texttt{4k Cycles}) \geq \tfrac{1}{4}n^2 + 1$.
\end{theorem}

\begin{proof}
	By \Cref{theorem:generator_labelling_minimalTC}, the generator labelling results in a minimal temporally connected graph. 
	By definition, this labelling assigns $\tfrac{n}{2}$ labels to edge $e$, which are from list $L_1$. Then it distributes $\tfrac{n}{2}-1$ labels to the next pair of clockwise/counter-clockwise edges, from list $L_2$. 
	Continuing, it assigns 1 label to each of the next pair of edges, from $L'_1$, and distributes $\tfrac{n}{2}-2$ labels on them from $L_1$, and then assigns 1 label to each of the next pair of edges, from $L'_2$, and distributes $\tfrac{n}{2}-3$ labels on them from $L_2$.
	This continues in this fashion, assigning one more label to each of the pair of edges as lists $L'_1$ and $L'_2$ grow, and distributing two less labels as lists $L_1$ and $L_2$ deplete. Note that this essentially adds two labels and removes two labels each time, so this remains a total quantity of $\frac{n}{2}$ and $\frac{n}{2}-1$ \textit{resp.} on each pair of edges.
	Finally, on edge $e'$, the generator labelling assigns $\tfrac{n}{4}+1$ labels, which in total gives:
	\begin{align*} 
		T = &\frac{n}{4} + 1 + (\frac{n}{2} + \frac{n}{2}-1 + \frac{n}{2} + \frac{n}{2}-1 + ... +\frac{n}{2} + \frac{n}{2}-1) \hfill\texttt{\qquad repeating $\frac{n}{4}$ times}\\
		= &\frac{n}{4} + 1 + \frac{n}{4}(n-1)\\
		= &\frac{1}{4}n^2 + 1
	\end{align*}
\end{proof}

The temporal cost of the generator labelling is the highest presented in this paper, and with it also comes the highest temporality, namely on edge $e$.

\begin{corollary}
	\label{corollary:generator_labelling_tauplus}
	$\tau^+ \geq \tau^+(\texttt{4k Cycles}) \geq \tfrac{1}{2}n$.
\end{corollary}

We claim that the generator labelling produces a minimal temporally connected graph for all even $n$. This can be proven by either extending our proof so the inductive step is from $n$ to $n+2$, showing this extends dominating journeys by exactly one time edge instead of two. This adds some symmetry case analysis to our already non-trivial proof, which is why we choose to omit it. Another option is to keep the inductive step from $n$ to $n+4$ but start with a cycle graph $C_{4k+2}$ instead of $C_{4k}$, such as $C_6$, which avoids the symmetry case analysis. Analysis of the density as in \Cref{theorem:generator_labelling_Tplus} and \Cref{corollary:generator_labelling_tauplus} reveals that for all $n$ even, $\tau^+(\texttt{Even Cycles}) \geq \tfrac{1}{2}n$ as well but $T^+(\texttt{Even Cycles}) \geq \tfrac{1}{4}n^2 + 1$ which are the same results as for \texttt{4k Cycles}.\footnote{The trivial cycle of size $2$ is an exception: let the generator labelling add label 1 on the single edge.}

Also, we can adapt the generator labelling slightly, so that it works on odd cycles.

\begin{definition}[Generator labelling of odd cycle graph $G$]
	See also \Cref{fig:generator_labelling_odd}. 
	Let list $L_1$ initially contain all odd natural numbers between 1 and $n$ included, and list $L_2$ contain all even natural numbers between 1 and $n$ included. Let list $L'$ be initially empty. 
	Take some edge $e$ of $G$ and put all labels from $L_1$ on it. Remove the largest label from $L_1$. Now consider the clockwise incident edge $e_c$ and the counter-clockwise incident edge $e_{cc}$, and assign the smallest label from $L_2$ to $e_c$, the second smallest to $e_{cc}$, the third back to $e_c$, and so on, distributing the labels of $L_2$ to edges $e_c$ and $e_{cc}$ in this alternating manner. Remove the largest label from $L_2$, and let $e_c$ be the next clockwise edge, and $e_{cc}$ the next counter-clockwise edge.
	Now, repeat the following process. 
	Distribute the labels from $L'$ to edges $e{cc}$ and $e_c$ in the alternating manner starting on edge $e_{cc}$ with the smallest label. When all labels from $L'$ have been distributed, continue the alternating distribution with the labels from $L_1$. Finally, move the smallest label from $L_1$ to $L'$ and remove the largest label from $L_1$, and let edge $e_c$ be the next clockwise edge and edge $e_{cc}$ be the next counter-clockwise edge. 
	Repeat this process, by using list $L_2$ next and starting the distribution on $e_c$, and then repeat it again using list $L_1$ and starting on $e_{cc}$ again, and so forth alternating between these pairs of lists and edges $e_c$ and $e_{cc}$. 
	The last time this process is repeated is when $e_c$ and $e_{cc}$ are incident edges.
\end{definition}

\begin{figure}[h]
	\begin{subfigure}{.5\textwidth}
		\includegraphics[width=.88\textwidth]{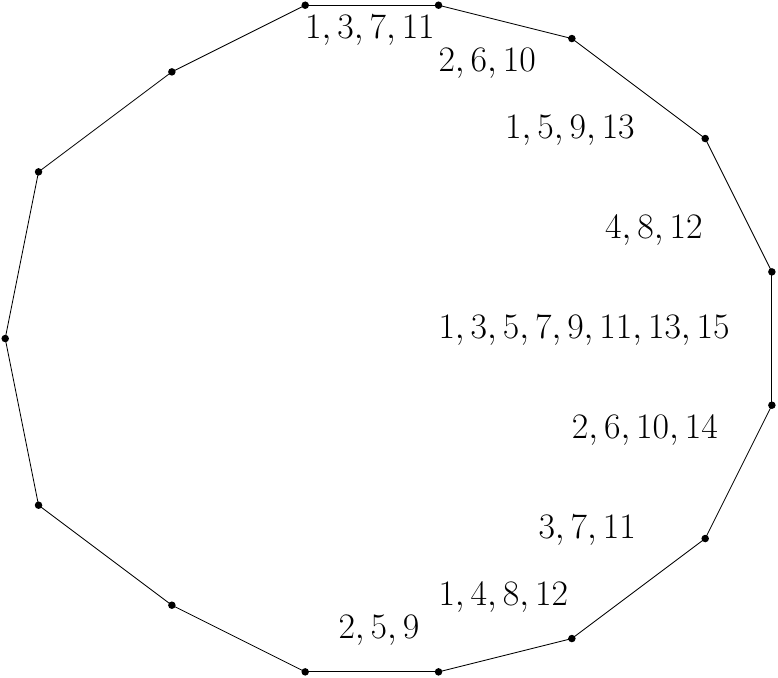}
		\caption{Start on $e_{cc}$, $L' = (1,2)$ and $L_1 = (3,5,7,9,11)$.}
	\end{subfigure}
	\hfill
	\begin{subfigure}{.5\textwidth}
		\includegraphics[width=.88\textwidth]{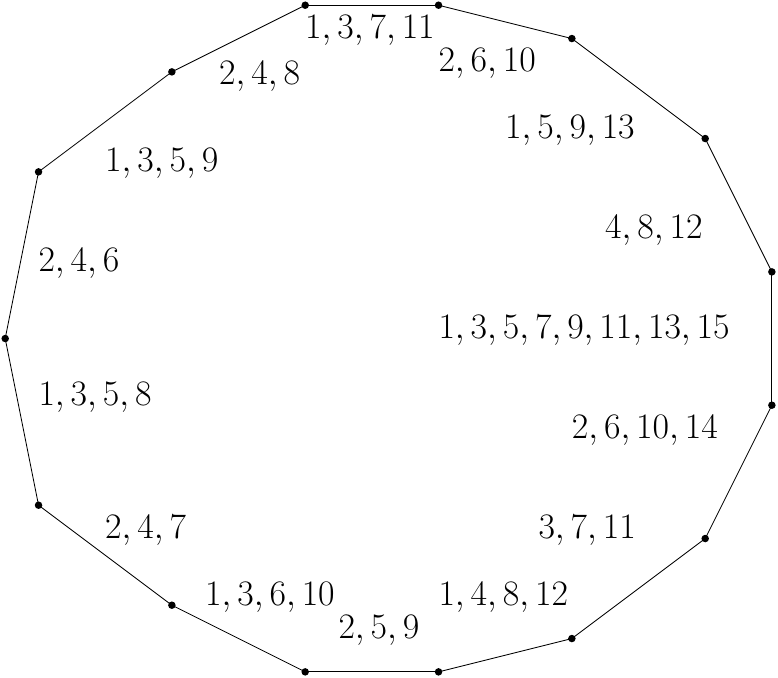}
		\caption{Start on $e_c$, $L' = (1,2,3,4,5)$ and $L_2 = (6,8)$.}
	\end{subfigure}
	\begin{subfigure}{.5\textwidth}
		\vfill
		\includegraphics[width=1\textwidth]{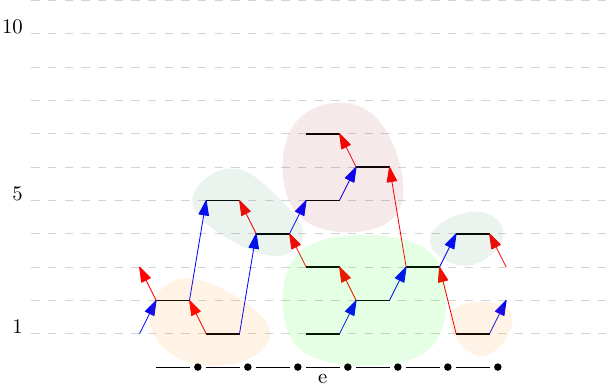}
		\caption{Base case cycle graph $C_7$.}
	\end{subfigure}
	\hfill
	\begin{subfigure}{.5\textwidth}
		\includegraphics[width=1\textwidth]{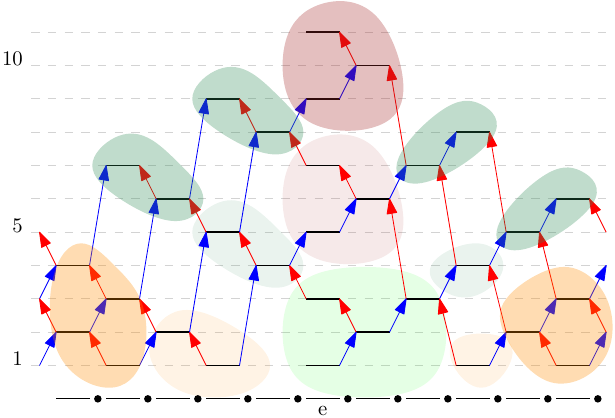}
		\caption{Inductive step extension to $C_{n+4}$ (here $C_{11}$).}
	\end{subfigure}
	\caption{The generator labelling for odd cycle graphs: \textbf{(a, b)} The process halfway through and at the final step respectively for cycle graph $C_{15}$; \textbf{(c, d)} The link stream representation of base case $C_7$ with seed, trunk, branches, and base, which is extended to $C_{11}$ with apex, leaves, and roots. 
	}
	\label{fig:generator_labelling_odd}
\end{figure}

This labelling for odd cycles can be proven minimal and temporally connected in a similar manner as for even cycles (see \Cref{fig:generator_labelling_odd} for the base case and inductive step). Again, a similar analysis as for \Cref{theorem:generator_labelling_Tplus} gives corresponding density results: 

\begin{theorem}
	\label{theorem:generator_labelling_odd_Tplus}
	$T^+(\texttt{Odd Cycles}) \geq 
	\tfrac{1}{4}n^2 + \tfrac{3}{4}$.
\end{theorem}

\begin{theorem}
	\label{theorem:generator_labelling_odd_tauplus}
	$\tau^+ \geq \tau^+(\texttt{Odd Cycles}) \geq \tfrac{1}{2}n + \tfrac{1}{2}$.
\end{theorem}

Putting all these densities of specific classes together, we thus obtain the following densities for the general class of cycle graphs.

\begin{theorem}
	\label{theorem:cycles_Tplus}
	$T^+(\texttt{Cycles}) \geq \tfrac{1}{4}n^2 + \tfrac{3}{4}$.
\end{theorem}

\begin{proof}
	All cycles attain a temporal cost of at least $\tfrac{1}{4}n^2$ through the generator labelling.
\end{proof}

\begin{theorem}
	\label{theorem:cycles_tauplus}
	$\tau^+(\texttt{Cycles}) \geq \lceil \tfrac{1}{2}n \rceil$.
\end{theorem}

\begin{proof}
	All cycles attain a temporality of at least $\lceil \tfrac{1}{2}n \rceil$ through the generator labelling.
\end{proof}

\section{Cactus graphs}
\label{sec:cacti}

As a reminder, cactus graphs are graphs such that any edge is part of at most one simple cycle. 
Cactus graphs are a superclass of both \texttt{Trees} and \texttt{Cycles}, so both classes give upper bounds on densities for \texttt{Cacti}. Moreover, \texttt{Cacti} subgraph dominates \texttt{Trees}, meaning $T^+(\texttt{Cacti}) = T^+(\texttt{Trees}) = 2n-3$ and $\tau^+(\texttt{Cacti}) = \tau^+(\texttt{Trees}) = 2$. These densities are however very unrepresentative of cactus graphs which are not trees. For this reason, we explore cactus graphs a bit more in detail by considering the additional parameter $c$, which is the largest simple cycle size, or circumference, of the graph. In other words, we technically study the classes \texttt{Circumference $c$ Cacti} in this section, but for simplicity we refer to these as \texttt{Cacti} in the following. Analysing cactus graphs with circumference $c$ combines the labellings and results from \texttt{Trees} (when $c=2$) and \texttt{Cycles} (when $c=n$). 


\begin{definition}[Combined labelling of cactus graph $G$]
	See also \Cref{fig:combined_labelling}. 
	Start by identifying the largest simple cycle in $G$, denote it as $C$ and note $|C| = c$. Cycle $C$ is contracted into a supervertex $v_C$ for the first and last step of this labelling, and is the focus of the second step.
	The first step starts the pivot labelling on some spanning tree with the pivot vertex $p = v_C$, and halts after assigning label $n-c$ on edge $\{v', v_C\}$.
	Step two consists of applying the generator labelling on $C$. Shift the resulting labels by adding $n-c$. Let $L$ denote the maximum label used up to now.
	Finally, in step three, continue and finish the pivot labelling but do not remove the largest label on edge $\{v', v_C\}$ if $c > 2$. Shift these last labels by adding $L+1$.
\end{definition}

Note that the combined labelling reverts to the pivot labelling if $c = 2$, \textit{i.e.} if $G$ is a tree, and to the generator labelling if $c = n$, \textit{i.e.} if $G$ is a cycle.

\begin{figure}[h]
	\begin{subfigure}{.5\textwidth}
		\includegraphics[width=.95\textwidth]{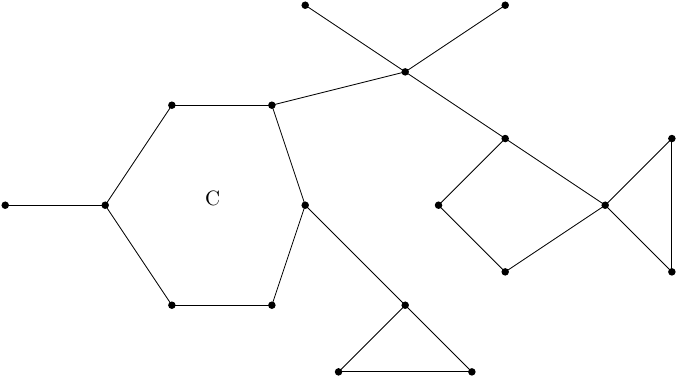}
		\caption{Cactus graph $G$, with largest cycle $C$ on the left.}
	\end{subfigure}
	\hfill
	\begin{subfigure}{.5\textwidth}
		\includegraphics[width=.95\textwidth]{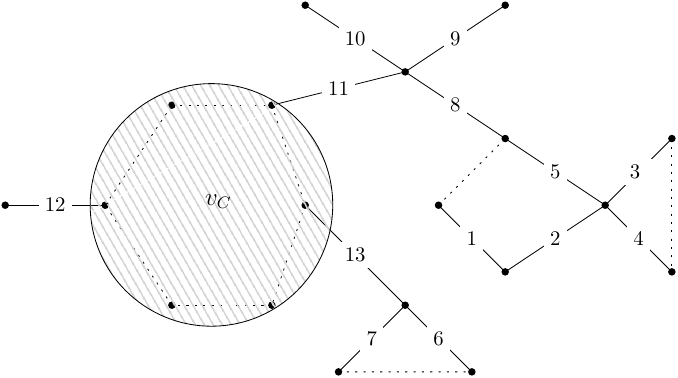}
		\caption{First step: converge to $v_C$ with pivot labelling.}
	\end{subfigure}
	\begin{subfigure}{.5\textwidth}
		\vfill
		\includegraphics[width=.95\textwidth]{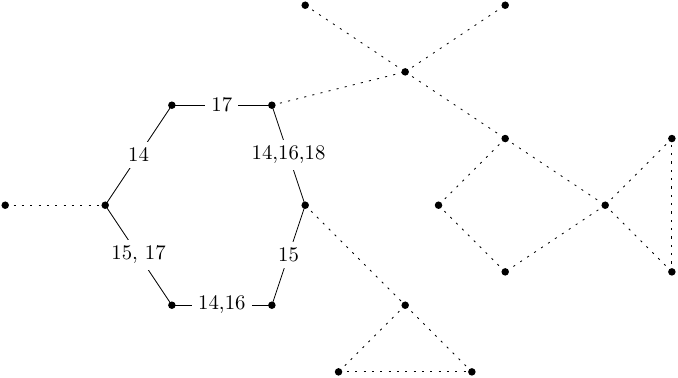}
		\caption{Second step: generator labelling on $C$.}
	\end{subfigure}
	\hfill
	\begin{subfigure}{.5\textwidth}
		\includegraphics[width=.95\textwidth]{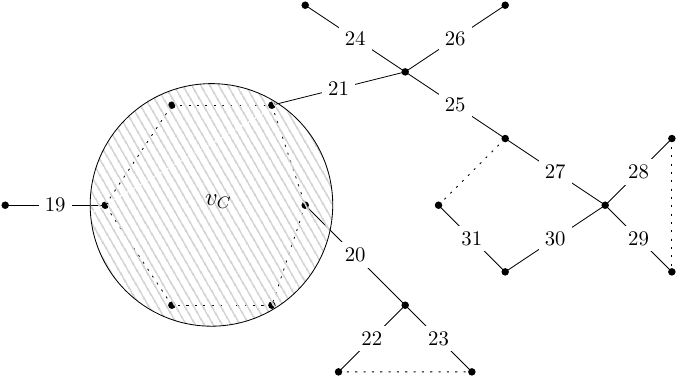}
		\caption{Third step: disperse from $v_C$ with pivot labelling.}
	\end{subfigure}
	\caption{Example of the combined labelling on cactus graph $G$. Step one and three consider the cycle $C$ to be contracted into supervertex $v_C$. Dotted edges are disregarded in corresponding steps.
	}
	\label{fig:combined_labelling}
\end{figure}

\begin{theorem}
	\label{theorem:combined_labelling_minimalTC}
	The combined labelling yields a minimal temporally connected graph.
\end{theorem}

\begin{proof}
	Cases $c=2$ and $c=n$ both are proven to yield a minimal temporally connected graph already in \Cref{sec:trees} and \Cref{theorem:generator_labelling_minimalTC} respectively. 
	Consider in the following $2 < c < n$.
	The labels from steps one and three are necessary, since together they apply the pivot labelling, which is proven to result in a minimal temporally connected graph. The exception to this is the largest label on edge $\{v', v_C\}$ which is removed in the pivot labelling but is not in the combined labelling. Indeed, in the pivot labelling, pivot vertex $p$ is reached by all vertices at time $n-1$ and thus can start reaching all vertices immediately at the same date. In this labelling, supervertex $v_C$ is reached by all vertices as well at time $n-c$, but this is not the case for the actual vertex in $C$ which is adjacent to $v'$. It is however reached by all vertices after step two, or after time $L$, and therefore the two labels on edge $\{v', v_C\}$ are necessary. 
	All labels used in step two are also necessary, since they ensure reachability among the vertices in $C$ by \Cref{theorem:generator_labelling_minimalTC}, which step one and three cannot possibly ensure.
	The resulting temporal graph is thus minimal. 
	It is also temporally connected, since any vertex $u$ can reach any other vertex $v$ by simply following the corresponding path along the spanning tree. If this path goes through the cycle $C$, then the vertex it arrives on in $C$ can reach any other vertex in $C$, thus also the vertex adjacent to the remainder of the path, after which it can continue the path. 
\end{proof}

For cactus graphs with circumference $c$, we thus obtain the following densities.

\begin{theorem}
	\label{theorem:cacti_Tplus}
	$T^+(\texttt{Cacti}) \geq \tfrac{1}{4}c^2 + 2(n-c)$.
\end{theorem}

\begin{proof}
	If $c=2$, the temporal cost is $2n-3$ (\Cref{theorem:trees_Tplus_eq_2n3}) and if $c=n$, the temporal cost is $\tfrac{1}{4}n + \tfrac{3}{4}$ (\Cref{theorem:cycles_Tplus}). Consider now $2 < c < n$. Step one adds exactly $n-c$ labels. Step two, by \Cref{theorem:cycles_Tplus}, uses at least $\tfrac{1}{4}c^2 + \tfrac{3}{4}$ labels. Finally, step three, being symmetric to step one, uses $n-c$ labels as well. This results in at least $\tfrac{1}{4}c^2 + 2(n-c) + \tfrac{3}{4}$ labels, which correlates with the case $c = n$, but does not for $c=2$, as this results in a cost greater than $2n-3$, which by \Cref{theorem:trees_Tplus_eq_2n3} is impossible. Removing the constant term $\tfrac{3}{4}$ resolves this issue; in other words, all cactus graphs can attain a density of at least $\tfrac{1}{4}c^2 + 2(n-c)$. 
\end{proof}


The densest edge in the combined labelling is the edge on which step two starts the generator labelling in cycle $C$ (unless $c=2$ in which case the temporality is 2).

\begin{corollary}
	\label{corollary:cacti_tauplus}
	$\tau^+(\texttt{Cacti}) \geq \lceil \tfrac{1}{2}c \rceil$.
\end{corollary}

Note that for $c=2$, these results are proven tight in \Cref{sec:trees}. The tightness of these results for other $c$ is open, but due to \Cref{lemma:bridge_edge},  we know that dense temporal cactus graphs have a density depending mainly on the cycles, \textit{i.e.} the structure outside of the cycles are all bridge edges and can thus only contribute a linear amount of labels.

\section{Conclusion}
\label{sec:conclusion}

Our results on densities of different graph classes are presented in \Cref{fig:results}. Since density results of a class \texttt{Class} transfer to any superclass \texttt{Superclass} which subgraph dominates \texttt{Class}, our (tight or lower bound) results thus transfer for \texttt{Connected}, the class of connected graphs, \texttt{Hamiltonian}, the class of Hamiltonian graphs, and \texttt{Circumference $c$}, the class of graphs of circumference $c$. 

\begin{figure}
	\centering
	\begin{tabular}{|c || c | c | c |}
		\hline
		Measure \textbackslash~Graph class          & \texttt{Trees} (\texttt{Connected}) & \texttt{Cycles} (\texttt{Hamiltonian}) & \texttt{Cacti} (\texttt{Circumference $c$})\\
		\hline
		\hline
		Maximum temporal cost $T^+$ & \makecell{$2n-3$ \\ \Cref{theorem:trees_Tplus_eq_2n3}} & \makecell{$\geq \tfrac{1}{4}n^2 + \tfrac{3}{4}$ \\ \Cref{theorem:cycles_Tplus}} & \makecell{$\geq \tfrac{1}{4}c^2 + 2(n-c)$ \\ \Cref{theorem:cacti_Tplus}} \\
		\hline
		Maximum temporality $\tau^+$ & \makecell{$2$ \\ \Cref{theorem:trees_tauplus}} & \makecell{$\geq \lceil \tfrac{1}{2}n \rceil$ \\ \Cref{theorem:cycles_tauplus}} & \makecell{$\geq \lceil \tfrac{1}{2}c \rceil$ \\ \Cref{corollary:cacti_tauplus}} \\
		\hline
	\end{tabular}
	\caption{Main results of our density measures on specific classes of graphs. Lower bounds are marked with $\geq$. 
		\label{fig:results}}
\end{figure}

We believe it may be promising to try and show that the lower bounds for cycles are optimal, by lowering the presented upper bounds, especially concerning temporality which may be easier to lower than temporal cost. Indeed, through our experiments with PTGen, no cycle graph has admitted a larger labelling than the generator labelling, leading us to believe that for some reason the temporality cannot exceed $\lceil \tfrac{1}{2}n \rceil$.


Note that in our paper, we specifically focussed on proper labellings. Other types of labellings exist (see \textit{e.g.} \cite{casteigts2022simple}) such as happy labellings which are proper labellings with at most one label per edge, or strict labellings which allow incident edges to have the same label(s) but do not allow journeys to cross multiple edges with a same label. When considering these different types of labellings, and the corresponding densest labellings possible, we obtain the preliminary results presented in  \Cref{fig:other_labellings}.

\begin{figure}
	\centering
	\begin{tabular}{|c || c | c | c | c |}
		\hline
		Measure \textbackslash~Labelling & Proper & Happy & Strict \\
		\hline
		\hline
		Maximum temporal cost $T^+$ & \makecell{$\geq \tfrac{1}{4}n^2 + 1$\\ \Cref{theorem:generator_labelling_Tplus}} & \makecell{$\geq \tfrac{1}{18}n^2 + \tfrac{35}{18}n - O(1)$\\ \Cref{theorem:happy_labelling_Tplus}} & \makecell{$\geq \tfrac{1}{2}n^2 - \tfrac{1}{2}n$\\ $\mathcal{G} = (G_1 = K_n)$}\\
		\hline
		Maximum temporality $\tau^+$ & \makecell{$\geq \tfrac{1}{2}n + \tfrac{1}{2}$\\ \Cref{theorem:generator_labelling_odd_tauplus}} & \makecell{$1$\\ by definition} & \makecell{$> \frac{1}{2}n + \tfrac{1}{2}$ ?\\experimentally} \\
		\hline
	\end{tabular}		
	\caption{\label{fig:other_labellings} Comparison of densest temporal graphs depending on the types of labellings considered.}
\end{figure}

Of course, the results for proper labellings were presented in this paper. Note that although both densities were obtained from analysing cycle graphs, the maximum temporal cost $T^+$ was specifically from \texttt{Even Cycles}, and the maximum temporality $\tau^+$ specifically from \texttt{Odd Cycles}.
We would like to note as well that we ran our labelling generator on general graphs instead of constraining it to cycle graphs only. Empirical evidence, for general temporal graphs up to size $n=8$, suggested that maximum temporality $\tau^+ = \lceil \tfrac{1}{2}n \rceil$ may be optimal not only in cycle graphs, but in general graphs as well, which would indicate that proper labellings may not be ideal for creating the densest temporal graphs (especially when compared to strict labellings, discussed below). These optimality results, for cycles and for general graphs, unfortunately remain to be proven.

The labelling from Axiotis and Fotakis is a happy labelling, and should thus take the spot for densest happy labelling. We have shown however in \Cref{corollary:adhoc_construction_Tplus} that our ad-hoc construction from \Cref{sec:better_lower_bounds} is a slightly denser labelling. The problem is that this labelling was specifically designed to attain a large maximum temporality, and is thus not a happy labelling. We show below that we can transform our labelling into a happy labelling, while still beating the one from Axiotis and Fotakis, which as a reminder is of size $\tfrac{1}{18} n^2 + \tfrac{3}{2}n + O(1)$.

\begin{theorem}
	\label{theorem:happy_labelling_Tplus}
	Concerning happy labellings, $T^+ \geq \tfrac{1}{18}n^2 + \tfrac{35}{18}n - O(1)$.
\end{theorem}

\begin{proof}
	Take the ad-hoc labelling, which is a happy labelling except for edge $e$ which has $k$ labels. To modify this labelling so as to obtain a happy one, one idea is to replace each of these labels $ik^2$ by introducing an intermediary vertex with edge and label $ik^2$ to vertex $a$ and edge with label $ik^2 + \epsilon$ to vertex $b$ for some $\epsilon < 1$ (and afterwards reordering the labels so they are all integers). This trick is similar to the so-called ``semaphores'' used in \cite{casteigts2022simple, bhadra2012computing, balev4590651temporally}. Supposing we could solve additional minor issues such as temporal connectivity concerning these additional vertices, this would add a linear amount of vertices to the construction, which would ultimately make the temporal cost lower than Axiotis and Fotakis' construction. 
	
	Nevertheless, our adaptation follows this idea, but instead of adding new intermediary vertices, we use vertices which are already present, being vertices $v_i$. Indeed, vertices $w_i$ are not suitable to act as the intermediary vertices, as the edges $\{u_i, w_j\}$ would interfere with the new journeys. Vertices $v_i$ do not have this problem, but introduce another issue: vertices $v_i$ already have edges to vertex $b$. We solve this by splitting up vertex $b$ into two vertices $b_1$ and $b_2$. See also \Cref{fig:adhoc_adaptation}.
	We thus add these new journeys replacing the multiple labels on edge $e$, alternating between vertices $b_1$ and $b_2$, and previous edges $\{b, v_i\}$ alternating as well. Due to another minor issue, we use vertex $v_{i+1}$ for journey with label $ik^2$. This leaves label $k^3$ on edge $e = \{a, b_1\}$, which is acceptable as it is the only label on that edge, respecting the constraints of a happy labelling. Finally, to ensure temporal connectivity (particularly from and to new vertices $b_1$ and $b_2$), we add edge $\{b_1, b_2\}$ with label 1 and edge $\{b_2, u_k\}$ with label $k^4$. 
	
	\begin{figure}[h]
		\begin{subfigure}{.5\textwidth}
			\includegraphics[width=\textwidth]{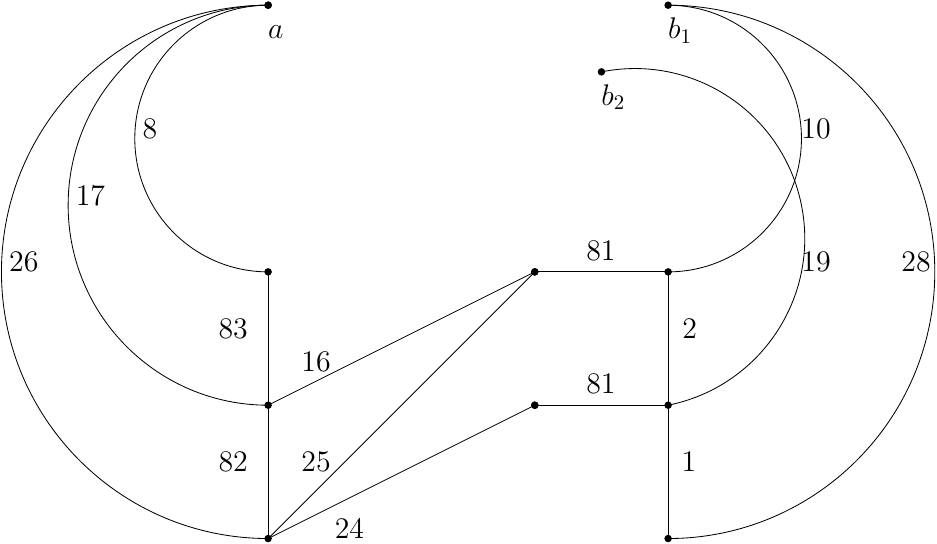}
			\caption{Original construction with vertex $b$ split into two.}
		\end{subfigure}
		\hfill
		\begin{subfigure}{.5\textwidth}
			\includegraphics[width=\textwidth]{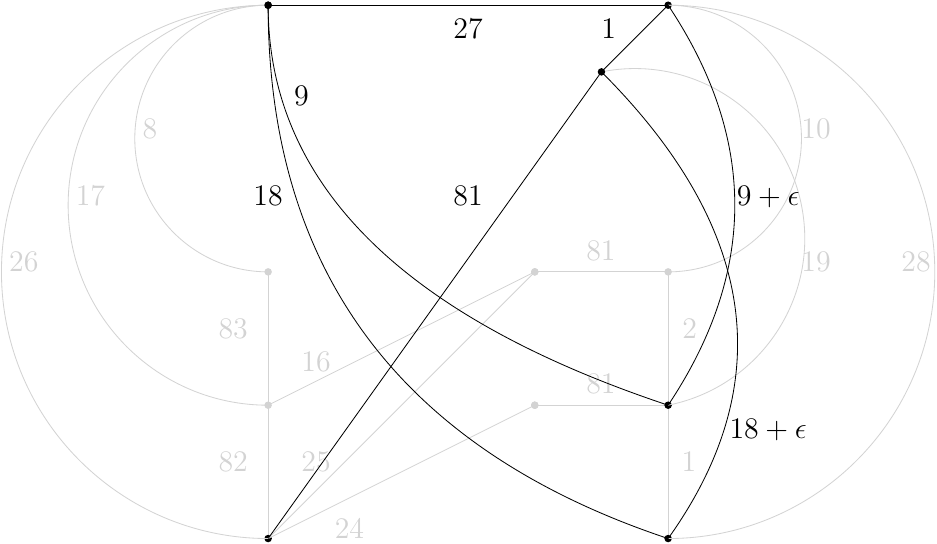}
			\caption{Journeys $u_i \leadsto v_i$ now go through vertex $v_{i+1}$.}
		\end{subfigure}
		\caption{The ad-hoc construction with $k=3$ adapted for happy labellings.
		}
		\label{fig:adhoc_adaptation}
	\end{figure}
	
	The result is clearly a happy labelling. 
		Let us show the ad-hoc construction is temporally connected. All vertices can trivially reach their neighbours, so we only look at the vertices which aren't neighbours.
		
		\begin{itemize}
			\item Vertex $a$ can reach vertices $b_2$, $v_i$ and $w_i$ by journey $(a, v_{i+1}, b_1$ or $b_2, v_i, w_i)$. 
			\item Vertex $b_1$ and $b_2$ have the same reachability since they are joined by an edge with lowest label $1$. They can reach vertices $w_i$ through journey $(b_1$ or $b_2, v_i, w_i)$, and vertices $u_i$ through journey $(b_2, u_k, u_{k-1}, u_{k-2}, ..., u_i)$. 
			\item Vertices $u_i$ can reach vertex $b_1$ through journey $(u_i, a, b_1)$ and vertex $b_2$ through journey $(u_i, a, u_k, b_2)$. They can reach vertices $v_j$ and $w_j$ such that $i \leq j$ through journey $(u_i, a, v_{i+1}, b_1$ or $b_2, v_j, w_j)$.\footnote{Vertex $u_k$ is an exception: it can reach $v_k$ through journey $(u_k, a, b_1, v_k)$.} They can reach the other vertices $v_j$ and $w_j$ such that $i > j$ by journey $(u_i, w_j, v_j)$. Finally, vertices $u_i$ can reach all vertices $u_j$ such that $i > j$ through journey $(u_i, u_{i-1}, u_{i-2}, ..., u_j)$, and vertices $u_j$ such that $i<j$ through journey $(u_i, a, u_j)$.
			\item Vertex $v_1$ can reach vertices $a$, $b_1$ and $b_2$ by passing through vertex $v_2$. Vertices $v_i$ can reach vertices $u_j$ through journey $\{v_i, a, u_k, u_{k-1}, u_{k-2}, ..., u_j\}$. 
			Vertices $v_i$ can reach vertices $v_j$ and $w_j$ such that $i \geq j$ through journey $(v_i, v_{i-1}, v_{i-2}, ..., v_j, w_j)$, vertex $w_{i+1}$ through journey $(v_i, v_{i+1}, w_{i+1})$, and other vertices $v_j$ and $w_j$ through journey $(v_i, b_1$ or $b_2, v_j, w_j)$.
			\item Finally, vertices $w_i$ can reach vertices $u_j$ through journey $(w_i, u_k, u_{k-1}, u_{k-2}, ..., u_j)$.
			They can reach vertices $a$ and $b_1$ through journey $(w_i, u_k, a, b_1)$, and vertex $b_2$ through $(w_i, u_k, b_2)$. 
			Vertices $w_i$ can reach vertices $v_j$ and $w_j$ such that $i < j$ through journey $(w_i, u_j, a, v_{j+1}, b_1$ or $b_2, v_j, w_j)$. Finally, they can reach vertices $v_j$ and $w_j$ such that $i > j$ through journey $(w_i, u_k, w_j, v_j)$.
		\end{itemize}
		Now let us prove the labelling is minimal.
		\begin{itemize}
			\item Edges $\{u_i, a\}$ are necessary for $u_i$ to reach $v_i$.
			\item Edges $\{a, v_i\}$ are necessary for $u_{i-1}$ to reach $v_{i-1}$.
			\item Edges $\{v_i, b_1$ or $b_2\}$ are necessary for $u_{i-1}$ to reach $v_{i-1}$, or for $u_i$ to reach $v_i$.
			\item Edge $\{a, b_1\}$ is necessary for vertex $u_k$ to reach vertex $v_k$.
			\item Edges $\{u_i, u_j\}$ are necessary for $u_k$ to reach $u_1$.
			\item Edges $\{v_i, v_j\}$ are necessary for $v_k$ to reach $v_1$.
			\item Edges $\{v_i, w_i\}$ are necessary for $v_i$ to reach $w_i$.
			\item Edges $\{u_i, w_j\}$ are necessary for $u_i$ to reach $w_j$.
			\item Edge $\{b_1, b_2\}$ is necessary for $b_1$ to reach $u_k$.
			\item Edge $\{b_2, u_k\}$ is necessary for $b_2$ to reach $u_k$.
		\end{itemize}
	
	Note that now $n = 3k +2$, and thus the total number of labels is:
	\begin{align*} 
		T&=k + k + k-1 + k-1 + k-1 + (1 + 2 + 3 + ... + k-1) + 2(k-1) + 3\\
		&= 7k -2 + \frac{k(k-1)}{2}\\
		&= \frac{1}{2}k^2 + \frac{13}{2}k -2\\
		&= \frac{1}{2}(\frac{1}{3}n - \frac{2}{3})^2 + \frac{13}{2}(\frac{1}{3}n - \frac{2}{3}) -2\\
		&= \frac{1}{18}n^2 + \frac{35}{18}n - \frac{55}{9}
	\end{align*}
\end{proof}

Concerning the results related to strict labellings, simply take a complete graph and assign label 1 to each edge to obtain a temporal cost of $\tfrac{1}{2}n^2 - \tfrac{1}{2}n$. We have also modified our labelling generator so as to work for strict labellings, and while running it on cycles, we have found labellings seemingly having more than $\tfrac{1}{2}n$ labels on an edge (see \Cref{fig:strict_labellings}). 
\begin{figure}[h]
	\begin{subfigure}{.5\textwidth}
		\includegraphics[width=.65\textwidth]{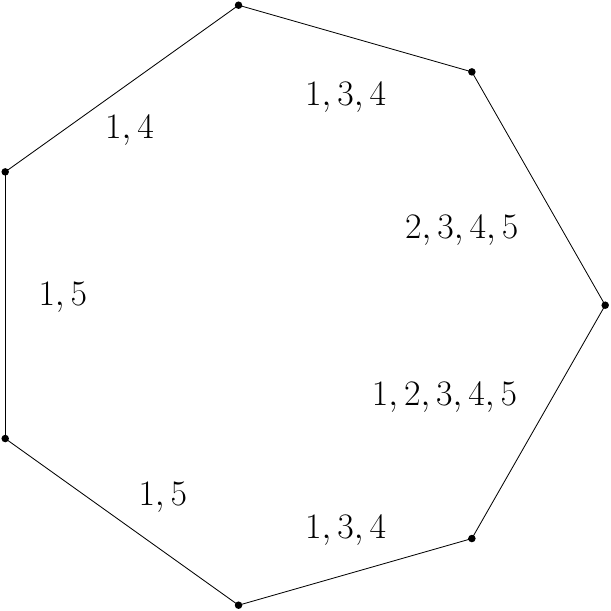}
		\caption{Cycle graph $C_7$ with temporality $\tau = 5$.}
	\end{subfigure}
	\hfill
	\begin{subfigure}{.5\textwidth}
		\includegraphics[width=.75\textwidth]{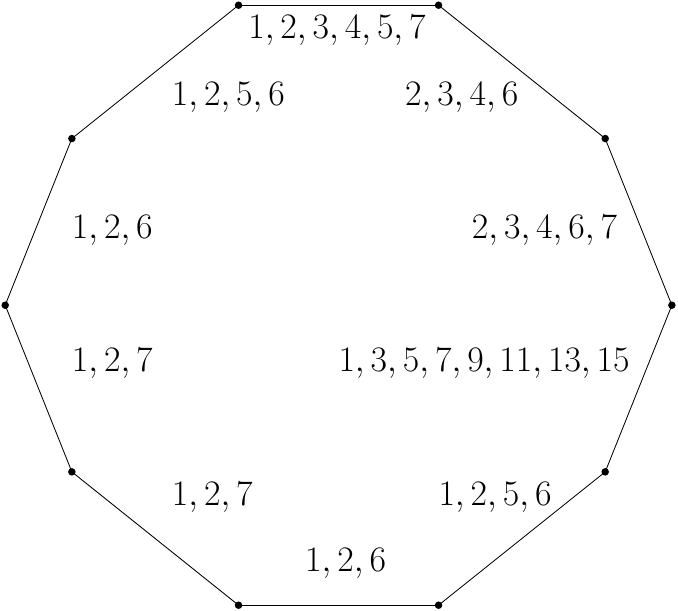}
		\caption{Cycle graph $C_{10}$ with temporality $\tau = 8$.}
	\end{subfigure}
	\caption{Minimal temporally connected graphs with empirically large $\tau$ in the strict setting.
	}
	\label{fig:strict_labellings}
\end{figure}
These examples seem to also have a larger temporal cost than $\tfrac{1}{4}n^2$, but do not empirically outperform the presented $\tfrac{1}{2}n^2 -\tfrac{1}{2}n$ (although the $C_7$ does attain it).
These are only experimental results. That being said, an intuition for strict labellings to allow for more labels than proper labellings can be the following. Consider the path graph on 3 vertices. Observe that to temporally connect this graph with a proper labelling, at most 3 labels can be used. However, with a strict labelling, 4 labels can be used. A more detailed analysis of this type of labellings is needed. 
We did not include non-strict labellings, which allow for journeys to use consecutive edges with a same label, as intuitively dense labellings would never use this property, as such a labelling could (again, intuitively) be made even denser by replacing any non-strict journey, which allows for journeys in both ways, with two strict journeys, essentially doubling the amount of labels. 
There is also the case of simple labellings, which are labellings with at most one label per edge. Of course, results from happy labellings transfer, but it is not directly clear how non-properness would improve upon happy labellings results.

Another interesting direction for future work is to study the computational complexity of corresponding decision problems, such as deciding if the maximum temporal cost or temporality of a given graph is at least some value $k$.
Indeed, even though some graphs admitting very dense labellings can be perceived as quite negative when considering the context of some adversary wasting precious networking resources, it may be computationally hard for the adversary to do so. The density measures presented for Hamiltonian graphs and for graphs of circumference $c$ correspond to labellings which are NP-hard to construct, as it requires knowledge of the Hamiltonian cycle, or of the simple cycle of size $c$.
If computing these density measures is proven to be NP-hard, then the study of polynomial cases, approximation algorithms, and fixed-parameter tractability algorithms could be of interest. 
Our results imply the corresponding problems are in P for trees and in APX for cycles (and cacti with large circumference) as the generator labelling is a 4-approximation for maximum temporal cost, and a 2-approximation for maximum temporality. However, any problem on cycle graphs is in the complexity class TALLY, and is thus unlikely to be NP-hard as this would imply $P=NP$.

\bibliography{paper}

\end{document}